\documentclass[12pt]{amsart}
\usepackage[margin=1in]{geometry}
\usepackage{newtxtext}
\usepackage{newtxmath}
\usepackage[alphabetic]{amsrefs}
\usepackage{graphicx}
\usepackage{caption}
\usepackage{tikz}
	\usetikzlibrary{decorations.pathreplacing}
	\usetikzlibrary{decorations.pathmorphing}
	\usetikzlibrary{patterns}
\usepackage{setspace}
\usepackage[hidelinks]{hyperref}

\allowdisplaybreaks

\newcommand{\cd}{\cdot}
\newcommand{\ra}{\rightarrow}
\newcommand{\pr}{\prime}

\newcommand{\te}{\theta}
\newcommand{\id}{\mathrm{id}}
\newcommand{\C}{\mathbb{C}}

\newcommand{\R}{\mathbb{R}}
\newcommand{\Z}{\mathbb{Z}}
\newcommand{\abs}[1]{\left\lvert #1 \right\rvert}
\newcommand{\tld}[1]{\widetilde{#1}}
\newcommand{\lbar}[1]{\overline{#1}}

\DeclareMathOperator{\E}{E}
\DeclareMathOperator{\pb}{P}
\DeclareMathOperator{\U}{U}

\numberwithin{equation}{section}
\numberwithin{figure}{section}

\newtheorem{theorem}{Theorem}[section]

\newtheorem{lemma}{Lemma}[section]
\newtheorem{proposition}{Proposition}[section]
\theoremstyle{definition}
\newtheorem{definition}{Definition}[section]
\theoremstyle{remark}
\newtheorem*{note}{Note}

\begin{document}

\title{Quantum Detection of Recurrent Dynamics}
\date{\today}

\author{Michael H. Freedman}
\address{\hspace{-\parindent}Michael H. Freedman}
\email{freedmanm@google.com, mfreedman@cmsa.fas.harvard.edu}

\begin{abstract}
	Quantum dynamics that explore an unexpectedly small fraction of Hilbert space is inherently interesting. Integrable systems, quantum scars, MBL, hidden tensor structures, and systems with gauge symmetries are examples. Beyond dimension and volume, spectral features such as an $O(1)$-density of periodic eigenvalues, or other spectral features, can also imply observable recurrence. Low volume dynamics will recur near its initial state $| \psi_0\rangle$ more rapidly, i.e.\ $\lVert\U^k | \psi_0\rangle - | \psi_0\rangle \rVert < \epsilon$, is more likely to occur for modest values of $k$, when the (forward) orbit $\operatorname{closure}(\{\U^k\}_{k=1,2,\dots})$ is of relatively low dimension $d$ and relatively small $d$-volume. We describe simple quantum algorithms to detect such approximate recurrence. Applications include detection of certain cases of hidden tensor factorizations $\U \cong V^\dagger(\U_1\otimes \cdots \otimes \U_n)V$. ``Hidden'' refers to an unknown conjugation, e.g.\ $\U_1 \otimes \cdots \otimes \U_v \ra V^\dagger(\U_1 \otimes \cdots \otimes \U_n)V$, which will obscure the low-volume nature of the dynamics. Hidden tensor structures have been observed to emerge both in a high energy context of operator-level spontaneous symmetry breaking \cites{fsz21,fsz21b,fsz21c,szbf23}, and at the opposite end of the intellectual world in linguistics \cites{smolensky90,mlds19}. We collect some observations on the computational difficulty of locating these structures and detecting related spectral information. A technical result, Appendix \ref{sec:qma-complete}, is that the language describing unitary circuits with no spectral gap (NUSG) around 1 is QMA-complete. Appendix \ref{sec:bridge} connects the Kolmogorov-Arnold representation theorem to hidden tensor structures.
\end{abstract}

\maketitle

\section{Introduction}

Our central theme is to identify structures, tensor and spectral, of unitary dynamics that will be detectable with NISQ-era quantum computers. Parallel structures arise in the wholly classical context of linguistics and machine learning (references in text), e.g.\ $\{\text{eigenvalues}\} \leftrightarrow \{\text{singular}$ $\text{values}\}$. We note these mathematical patterns and investigate one in an appendix (\ref{sec:bridge}).

Dynamical systems, quite generally, have associated dimensions and invariant measures. These can be defined abstractly \cite{shalizi06} or extrinsically after a Takens embedding \cite{takens06} places the system inside a Euclidean space $\R^N$. This holds for both continuous and discrete time dynamics. When an orbit explores a low-dimensional or low-volume portion of phase space, this is worth knowing; something special is happening. Examples include KAM tori, strange attractors, and in this note, the forward orbit $\{\U^k | \psi_0 \rangle\}_{k=1,2,\dots}$ for certain well-aligned pairs $(\U,| \psi_0 \rangle)$. Low dimensionality/volume may reflect constraints, conserved quantities, or the lack of efficient couplings to the environment as in many-body localization (MBL) \cite{fan17}. MBL is mentioned because it highlights that a practical meaning to ``forward orbit'' should be given. MBL states are long-lived meta-stable states, but eventually lose their structure, so they must be studied with a timescale in mind. Since our interest here is application to mid-term quantum computers, where errors will build up, we must consider some bound on forward orbit. Literally taking $k \ra \infty$ would, in the presence of noise, always lead to ergodicity. It is our goal here to highlight regimes where non-ergodic dynamics can be detected (perhaps for $k \leq 10^6$). Thus, in what follows, the orbit $\{\U^k | \psi_0 \rangle\}_{k=1,2,\dots}$ should have some cutoff; correspondingly, the notion of Lie group closure of the orbit also needs some technical adjustment.\footnote{E.g.\ Lie group closure of the infinite orbit may be replaced with a Lie subgroup of Gromov-Hausdorff distance $< \epsilon$ from the finite orbit. We will not enter into such details here.}

We study $N \times N$ unitaries $\U$ acting on $\mathcal{H}^N$, where often a qubit decomposition expresses $N = 2^n$. $G := \operatorname{closure}(\{\U^k\}_{k=1,2,\dots})$ is an abelian Lie sub-group of $\U(N)$ with a single (dense) generator, $\U$. This forces $G \cong \Z_j \times T^d$ to be the product of a finite cyclic group $\Z_j := \Z / j\Z$, and a $d$-torus for $0 \leq d \leq N$. We are interested in the case when $d << N$, $j = O(1)$ and $d$-volume$(T^d) = O((2\pi)^d)$, the volume of the standard product tori. Our object of study is not $G$ directly but the orbit $\mathcal{O} \subset \mathcal{H}^N$ as it acts on $\mathcal{H}^N$ by moving a standard $\Z$-basis vector, say $\vert \vec{0} \rangle$, about the closed orbit $\lbar{\mathcal{O}} := G \vert \vec{0}\rangle$.

We will explore a toy example where large $S \subset \mathcal{H}$ arises from a hidden tensor decomposition into factors enjoying a special spectral structure.
\begin{equation}
	\U = V(\U_1 \otimes \cdots \otimes \U_j)V^\dagger,\quad \U_i \text{ an } n_i \times n_i \text{ unitary},\ \prod_{i=1}^j n_i = n \tag{$\ast$}.
\end{equation}
Conjugating $V$ cryptographically hides the special tensor structure of the rhs. Our algorithm will, however, still respond to it. The hidden tensor structure alone, without an additional spectral assumption, does speed up recurrence, but not enough to put the problem of detecting it in BQP.

Note that if $V = \id$, then neither $\U$ (nor its powers) will entangle all qubits, so entropy measures readily identify ($\ast$). With $V \neq \id$ the suggested algorithms will play a role and could also have relevance in other quantum mechanical contexts where small-volume dynamics is key: integrable systems, many-body localization (MBL), scars, and many-body scars.

Emergent tensor decompositions of this type have recently been found to arise from a symmetry breaking process in which an initial Hilbert space, without a tensor structure, finds one in order to minimize certain loss functions \cites{fsz21,fsz21b,fsz21c,szbf23}. Pursuing, now, analogies outside quantum mechanics, tree-like tensor structures have been proposed to model the emergent structure of natural language \cites{smolensky90,mlds19}. These papers introduce tensor-paired concepts, called \emph{roles} and \emph{fillers}. For a simple example, consider verbs and nouns: Jill threw the ball, diagrammed as $(J \otimes T) \otimes B$. A syntactic example is reproduced in Section \ref{sec:comp-dif}. The detection within trained LLMs of emergent tensor structures---ones not anticipated by grammarians, and likely less interpretable, is the subject of current research \cite{fbt24}. This latter work highlights substantial numerical obstacles to detecting hidden tensor structures (in that work the context is general, not necessarily unitary, matrices).

I would like to thank Alexei Kitaev, Yuri Lenski, Jarrod McClean, and Vadim Smelyanskiy for stimulating conversations on these topics. I also thank the anonymous referee for helpful advice and literature references.

\subsection{Outline and Preview}
Section \ref{sec:rec-alg} (Recurrence Algorithm) details the quantum algorithm used to detect recurrence of $\mathcal{O}$ when accompanied by the aforementioned special feature: an $O(1)$ fraction of low period eigenvalues. Evolution of a system with a flat band, e.g.\ the integral quantum Hall effect (IQHE), will have a large invariant subspace $S \subset \mathcal{H}$. This corresponds to Hilbert space directions (spanning $S$) which do not participate in the Hamiltonian, and will make recurrence more likely.

In Section \ref{sec:comp-dif} (Complexity and Circuit Obfuscation), we consider the computational difficulty of detecting the structures under study: hidden tensor structures and the prevalence of near-unity (or near-periodic) eigenvalues. We find a related problem hard: the language, no unitary spectral gap (NUSG), whose members construct a quantum circuit with an eigenvalue near 1, is proven to be QMA-complete. (The proof is deferred to Appendix \ref{sec:qma-complete}.) We consider evidence that the related problems of breaking \emph{obfuscation} protocols, intended to hide the presence of difficult-to-determine structures, may be difficult.

Section \ref{sec:hidden-structure} (Hidden Structure in Dynamically Produced Linear Algebra) reviews some examples of emergent tensor structures in disparate contexts---high energy, linguistics, and machine learning---in the hope of establishing useful analogies and common features.

Appendix \ref{sec:qma-complete} (No Unitary Spectral Gap is QMA-Complete) proves QMA-completeness of an elementary spectral property of a unitary. This is done to illustrate that without additional context, the spectral properties upon which detectable recurrence rests are difficult to ascertain.

Appendix \ref{sec:bridge} (A Curious Bridge Between Tensor Factoring and the Kolmogorov-Arnold Representation Theorem) develops a connection between the algebra of tensor factorization and the Kolmogorov-Arnold representation theorem. The obstruction to low-dimensional representations of multivariate functions ``Sternfeld arrays,'' can likewise be viewed as obstructions to constructing tensor factorizations.

\section{Recurrence Algorithm}\label{sec:rec-alg}

We suppose $\U: \mathcal{H}^N \ra \mathcal{H}^N$ and its dyadic powers $\U^{2^j}$ are available to us, for example, as when $\U$ encodes modular exponentiation in Shor's algorithm. We use a number register to represent dyadically any $0 \leq k \leq 2^j$ (given a size cutoff for $k$ as discussed in the introduction). Let $|\psi_0 \rangle = | \vec{0} \rangle$ be a basis vector of $\mathcal{H}^N$. This notation reflects the special case when $\mathcal{H}^N \cong (\C^2)^{\otimes n}$ has a fixed qubit decomposition. With the number register written on top and the $\mathcal{H}^N$ state line at the bottom, build the following circuit:

\begin{figure}[ht]
	\centering
	\begin{tikzpicture}[scale=1.3]
		\draw [decorate,decoration={brace,amplitude=5pt}] (-0.8,0.5) -- (-0.8,3.7);
		\node at (-1.7,2.3) {number};
		\node at (-1.7,1.9) {register};
		\node at (-1.7,0.2) {state};
		\node at (-1.7,-0.2) {register};
		
		\node at (-0.6,0) {$|0,...,0\rangle$};
		\node at (-0.4,0.75) {$|0\rangle$};
		\draw (0,0) -- (8,0);
		\draw (0,0.75) -- (1.5,0.75);
		\draw (1.5,0.5) -- (1.5,1) -- (2,1) -- (2,0.5) -- cycle;
		\node at (1.75,0.75) {H};
		\draw (2,0.75) -- (6,0.75);
		\draw (2.75,0.75) -- (2.75,0);
		\draw[fill=black] (2.75,0.75) circle (0.2ex);
		\node at (2.75,0) {$\times$};
		\node at (2.75,-0.4) {$\operatorname{CU}$};
	
		\node at (-0.4,1.5) {$|0\rangle$};
		\draw (0,1.5) -- (1.5,1.5);
		\draw (1.5,1.25) -- (1.5,1.75) -- (2,1.75) -- (2,1.25) -- cycle;
		\node at (1.75,1.5) {H};
		\draw (2,1.5) -- (6,1.5);
		\draw (3.5,1.5) -- (3.5,0);
		\draw[fill=black] (3.5,1.5) circle (0.2ex);
		\node at (3.5,0) {$\times$};
		\node at (3.5,-0.4) {$\operatorname{CU}^2$};
	
		\node at (-0.4,2.25) {$|0\rangle$};
		\draw (0,2.25) -- (1.5,2.25);
		\draw (1.5,2) -- (1.5,2.5) -- (2,2.5) -- (2,2) -- cycle;
		\node at (1.75,2.25) {H};
		\draw (2,2.25) -- (6,2.25);
		\draw (4.25,2.25) -- (4.25,0);
		\draw[fill=black] (4.25,2.25) circle (0.2ex);
		\node at (4.25,0) {$\times$};
		\node at (4.25,-0.4) {$\operatorname{CU}^4$};
	
		\node at (1.75,3) {$\vdots$};
	
		\node at (-0.4,3.5) {$|0\rangle$};
		\draw (0,3.5) -- (1.5,3.5);
		\draw (1.5,3.25) -- (1.5,3.75) -- (2,3.75) -- (2,3.25) -- cycle;
		\node at (1.75,3.5) {H};
		\draw (2,3.5) -- (6,3.5);
		\draw (5.3,3.5) -- (5.3,0);
		\draw[fill=black] (5.3,3.5) circle (0.2ex);
		\node at (5.3,0) {$\times$};
		\node at (5.5,-0.35) {$\operatorname{CU}^{2^{j-1}}$};

		\node at (4.8,-0.45) {$\cdots$};
	
		\draw[fill=white] (7,0.2) circle (0.4);
		\draw (7,0.2) -- (7.28,0.48);
		\draw (6.5,-0.2) -- (7.5,-0.2) -- (7.5,-0.5) -- (6.5,-0.5) -- cycle;
		\node at (7,1.2) {measure in};
		\node at (7,0.8) {$\Z$-basis};
	\end{tikzpicture}
	\caption{}\label{fig:number-register}
\end{figure}

That is, apply Hadamard (H) to each qubit of the number register, then apply controlled $\U^{2^i}$ from the $i$th number qubit to the state register. Finally, measure the state register, with particular attention to $\vert \vec{0} \rangle$ recurring as the measurement outcome. Note that this circuit is similar to the standard one for quantum phase estimation \cite{ksv02}, but differs in two ways. First, the technically demanding $(\text{Quantum Fourier Transform})^{-1}$, $\text{QFT}^{-1}$, is omitted. Second, we measure the state register, \emph{not} the number register.

By itself, the circuit in Figure \ref{fig:number-register} yields exactly the same probability of detecting (measuring $|\vec{0}\rangle$) as the more classical process of setting the control bits with fair coin flips, instead of Hadamard's. However, generating the superposition is advantageous if one possesses sufficient quantum computing resources to apply (a variant of) amplitude amplification \emph{before} the measurement step. This will quadratically enhance signal detection (see section \ref{sec:amplitude}). If even more quantum resources are available, after measuring the state vector, one may still apply $\text{QFT}^{-1}$ to the number register. This additional interference step harnesses the quantum phase estimation algorithm to assess the approximate period of any known approximate eigenvector.

\subsection{The case where \texorpdfstring{$\U$}{U} is Haar random}\label{sec:haar-random}
In this case, recurrence to $\vert \vec{0} \rangle$ under measurement is doubly exponentially rare. A truly Haar-random unitary must be regarded as emerging from a ``black box'' since it cannot come from a circuit of polynomial depth. For these, the recurrence algorithm cannot actually be implemented. However, to the extent that it is difficult to distinguish truly Haar-random from the output of a shallow circuit, our analysis of the Haar-random case applies.

Let $\U$ have eigenbasis $\{|\psi_1\rangle, \dots, |\psi_{2^n} \rangle$, with corresponding eigenvalues $\lambda_i$, $\U|\psi_i\rangle = \lambda_i | \psi_i \rangle$, and write $\lambda_i = e^{2\pi i \te_i}$. In the (unrelated) $\Z$-basis, writing $\mathcal{H}^N \cong \C^N \cong (C^2)^{\otimes n}$, we may express the all $|0\rangle$ basis vector, $| \vec{0} \rangle =: | \psi_0 \rangle$ as $| \psi_0 \rangle = \sum_{i=1}^{2^n} a_i \lambda_i$. With this notation, the overlap of $|\psi_0 \rangle := | \vec{0} \rangle$ with $\U^k|\psi_0\rangle$ is given by:
\begin{equation}\label{eq:haar-random}
	\langle \psi_0 |U^k| \psi_0 \rangle = \sum_{i=1}^{2^n} \abs{a_i}^2 \lambda_i^k.
\end{equation}

The Haar random assumption means the $2^n$-tuple $\vec{a} = (a_1,\dots,a_{2^n})$ is uniformly distributed, i.e.\ homogeneous under the action of the orthogonal group $O(2^n)$) on the unit sphere $S^{2^n-1}$ of $H^N$. The eigenvalues $\{\lambda_1^k,\dots,\lambda_{2^n}^k\}$ are Gaussian Unitary Ensemble (GUE-) distributed, and as $k$ becomes large, $k$ approaches uniform iid over $2^n$ phase circles $\U(1)$.

Thus, line \eqref{eq:haar-random} may be regarded as the formula for the result of a random walk with variable and weakly correlated step sizes $\abs{a_i}^2$. The result is similar to uniform random walks of $2^n$ steps, each of length $\frac{1}{2^n}$, taken in Euclidean $2^n$ space. If the step sizes $\{\abs{a_i}^2\}$ obey a relatively flat distribution, as they will when the $\Z$-basis is (Haar) randomly aligned to the eigenbasis $\{\vert\psi_i\rangle\}$, then for $n$ and $k$ larger than a single digit, the expected absolute value of the walk $\E\left[\abs{\sum_{i=1}^{2^n} \abs{a_i}^2 \lambda_i^k}\right]$ will very nearly approximate the rms value of uniform step sizes.
\begin{equation}\label{eq:rms-steps}
	\E\left[\abs{\sum_{i=1}^{2^n}\abs{a_i}^2 \lambda_i^k}\right] \approx 2^{-\frac{n}{2}},
\end{equation}
the familiar ``$\surd$-law.''

Fluctuations around this expected value are essentially Gaussian (increasingly so as $n$ grows). The probability that $\abs{\sum_{i=1}^{2^n} \abs{a_i}^2 \lambda_i^k}$ exceeds its expectation by a multiplicative factor $z \in \R^+$ is:
\begin{equation}\label{eq:mult-factor}
	\pb \left(\abs{\sum_{i=1}^{2^n} \abs{a_i}^2 \lambda_i^k} \geq z \cd \E\left[\abs{\sum_{i=1}^{2^n} \abs{a_i}^2 \lambda_i^k}\right]\right) \approx E_c(z)
\end{equation}
where $E_c(z)$ is the complementary error function which may be further approximated as
\begin{equation}
	E_c(z) \approx \frac{1}{6}e^{-z^2} + \frac{1}{2} e^{-\frac{4}{3}z^2},
\end{equation}
see \cite{error-function}.

The bottom line is that the Born rule probability to measure $|\vec{0} \rangle$ after $k$ applications of $\U$ to $| \vec{0} \rangle$, for a Haar random $\U$, is doubly exponentially suppressed. To observe recurrence, $z$ itself must be exponential in $k$.

Line \eqref{eq:mult-factor} relied on the Gaussian description of the rotationally-invariant spherical measure on $S^{k-1} \subset \R^k$, denoted as the random variable $X$ concentrated on $S^{k-1}$,
\begin{equation}
	X = (X_1,\dots,X_k) = \left(\frac{G_1}{G}, \dots, \frac{G_k}{G}\right)
\end{equation}
where each $G_i$, $1 \leq i \leq k$, is an independent copy of the standard Gaussian ($\mu = 0$, $\sigma^2 = 1$), and $G := \sqrt{\sum_{i=1}^k G_i^2}$.
\begin{equation}\label{eq:gaussian-approx}
	\begin{split}
		\pb\left(\abs{X_1} > \frac{z}{\sqrt{n}}\right) & = \pb\left(\frac{\abs{G_1}}{G} > \frac{z}{\sqrt{n}}\right) = \pb\left(\abs{G_1} > \frac{G}{\sqrt{n}}z\right) \\
		& \approx \pb(\abs{G_1} > z) = E_c(z),\ z = \sqrt{2}(\text{std devs})
	\end{split}
\end{equation}
where the only approximation in \ref{eq:gaussian-approx} results from replacing the random variable $G$ with its mean $\sqrt{n}$. Around this mean there will be $O(1)$ fluctuations.

\subsection{The case where \texorpdfstring{$\U$}{U} has a hidden tensor decomposition}\label{sec:tensor-decomp}
That is, $\U = V(\U_1 \otimes \cdots \otimes \U_n)V^\dagger$, all $\U_i$, $1 \leq i \leq n$, are $2 \times 2$ unitaries, with no additional structure assumed (still a negative example).

We essentially repeat the previous analysis, but now recognize the highly correlated nature of the eigenvalues $\Lambda_{\vec{j}}$ of $\U$, $\vec{j} = 1,\dots, 2^n$. Let $\U_i$ have eigenvalues $\lambda_{i,0}$ and $\lambda_{i,1}$, both unit complex. Then the eigenvalues of $\U$ are $\{\Lambda_{\vec{j}} := \prod_{i=1}^n \lambda_{i,j(i)},1 \leq i \leq n, j(i) \in \{0,1\}\}$. Writing $\lambda_{i,j} = e^{2\pi i \te_{i,j}}$, and $\Lambda_{\vec{j}}^k = \prod_{i=1}^n e^{ik\te_{i,j(i)}} = e^{i(\sum_{i=1}^n k\te_{i,j(i)})} =: e^{i\te_{\vec{j}}}$.

Now thinking of the $\te_{i,j}$ as iid uniform random variables on $[-\pi,\pi]$, let us consider the event where each of the $2n$ random variables $\te_{i,j}$ lies in $[-\frac{\pi\epsilon}{2},\frac{\pi\epsilon}{2}]$, an event with probability $\epsilon^{2n}$. This condition biases $\lambda_{i,j}$ strongly towards the positive real direction, $\R^+$. Tail events transmit this to a bias of $\Lambda_{\vec{j}}$ also towards $\R^+$, but this bias is exponentially weak for:
\begin{equation}\label{eq:exp-weak-bias}
	\epsilon > (\mathrm{const.}\sqrt{n})^{-1},\ \text{where the const.\ is } O(1) \text{ and } > 1.
\end{equation}
Line \eqref{eq:exp-weak-bias} follows from the random $\surd$-behavior producing $\Lambda_{\vec{j}}$ from $\{\lambda_{i,j(i)}\}$.

Thus, the induced bias towards recurrence is exponentially small, except for a $\left(\frac{1}{\mathrm{const.} \sqrt{n}}\right)^{2n}$ fraction of instances for $\{\U_i\}$, and thus will not be detectable unless $n$ is small enough to be classically simulated.

It is worth understanding why the tensor form just considered, with $2n$ independent parameters describing its eigenvectors, rather than the $2^n$ independent parameters needed to describe a Haar random $\U$, does not manifest robust recurrence for $n$ larger than about $16$. The reason lies in the fact that the number of states in the unit sphere $S^{2^n-1}$ with a fixed angular separation $< \frac{\pi}{2}$ is \emph{doubly} exponential in $n$ \cite{cs13}. Reducing the effective dof by $\log$ only removes \emph{one} of these exponentials. Angular separation $< \frac{\pi}{2} - \mathrm{const.}$ is efficiently detectable through (Born rule) measurement, whereas larger separation, angular separation $< \frac{\pi}{2} - \frac{\mathrm{const.}}{n^p}$, $p > 0$, is not. We note that with post-selection (Post BQP) the difference between single and doubly exponential case can be exploited to detect hidden factorizations; but this is not the world we live in.

\subsection{The case of a hidden tensor decomposition with structure}\label{sec:with-structure}
The additional \emph{structure} we consider is where the tensor factors have an unexpected abundance of eigenvalues of small finite order. For clarity, let us restrict to the case where the finite order eigenvalues are all $\lambda = 1$. This (up to an overall phase) may arise from the evolution of a flat band, IQHE-physics, and other topological contexts. A toy model for examples of spectral degeneracy are controlled rotations, controlled-controlled-rotations, etc.:
\begin{equation}
	C\te = \def\arraystretch{1}
	\begin{tikzpicture}[baseline=(current bounding box.center)]
		\node at (0,0) {$\left(\begin{array}{cccc}
			1 \\
			& 1 & & \text{\LARGE 0} \\
			\text{\LARGE 0} & & 1 \\
			& & & e^{2\pi i \te}
		\end{array}\right),\ CC\te = \left(\begin{array}{cccc}
			1 \\
			& \hphantom{n} & & \text{\LARGE 0} \\
			\text{\LARGE 0} & \hphantom{n} & \\
			& & & e^{2\pi i \te}
		\end{array}\right)$};
		\node at (1.8,0.65) {\footnotesize{1}};
		\node at (2,0.4) {\footnotesize{1}};
		\node at (2.2,0.15) {\footnotesize{1}};
		\node at (2.4,-0.1) {\footnotesize{1}};
		\node at (2.6,-0.35) {\footnotesize{1}};
		\node at (2.8,-0.6) {\footnotesize{1}};
	\end{tikzpicture}
\end{equation}

To take a fixed example to study, let us imagine that our state register consists of 72 qubits. We imagine a NISQ machine not yet capable of running phase estimation algorithms (the $\text{QFT}^{-1}$ being a hurdle), but capable of running the recurrence algorithm, representing numbers of a few hundred thousand in the number register.

We consider the problem of distinguishing a Haar random $\U$ from:
\[
	U^\pr = V(\U_1 \otimes \cdots \otimes \U_{24})V^\dagger,\ V \text{ Haar random},
\]
and each $\U_i$, $i = 1,\dots,24$ acts on its own 3 qubits as a $CC\te_i$, $\te_i \in [0,2\pi]$. For the $\te_i$ random, the fraction of the $2^{72}$ eigenvalues of $\U^\pr$ which are 1 is about $4\%$.
\begin{equation}
	\operatorname{frac}_1(\U^\pr) = \left(\frac{7}{8}\right)^{24} \approx 0.040569.
\end{equation}

Writing $\{\lambda_i, 1 \leq i \leq 2^{72}\}$ for the eigenvalues of $\U$ and referring back to \eqref{eq:haar-random}, we can replace \eqref{eq:rms-steps} with a revised expectation based on this fraction $\approx 4\%$. For this analysis, we assume no errors. The $\lambda_i^k$ not obviously 1 are uniformly distributed over the phase circle.
\begin{equation}
	\E\Big[\big\lvert \langle \vec{0} | \U^{\pr k} | \vec{0} \rangle\big\rvert\Big] = \E\left[\abs{\sum_{i=1}^{2^k} \abs{a_i}^2 \lambda_i^k}\right] = \text{bias-term} + \surd\text{-term} \approx 0.0405 \pm 0.9594\left(\frac{1}{\sqrt{2^{72}}}\right) \approx 0.0405.
\end{equation}
The $\surd$-term is entirely negligible.

The Born probability of $\Z$-basis measurement of $U^{\pr k} \langle \vec{0}|$ returning $\langle \vec{0}| \approx (0.0405)^2 \approx 0.001646 \approx (607.59)^{-1}$. If the original $\te_i$ are iid and uniform, this probability is independent of $k$. Consequently, the recurrence algorithm circuit must be run around 600 times to even have a chance of observing the spectral structure built into $\U^\pr$. Running the circuit 6000 times gives a 99.9\% chance of detecting the hidden structure if it is present. Of course, the chances of the measurement outcome being $|\vec{0}\rangle$ in the Haar random case, $\U^\pr$ replaced by $\U$, are 1 in $2^{72}$, vanishingly small.

\subsection{Caveats}
The form of the recurrence algorithm circuit requires that dyadic powers of $\U_i$ can be efficiently described. This is the case for our $U_i = CC\te_i$ example. If one generalizes to cases where powers of $\U_i$ do not have compact descriptions, the circuit must be replaced with brute force iteration, resulting in a length proportional to $k_{\text{max}}$ rather than $\log k_{\text{max}}$.

Up to this point we have made the unrealistic assumption that quantum operations are exact. If instead there is a uniform uncorrelated random error of the form
\begin{equation}
	\E\abs{|\psi\rangle \mathcal{O}_{\text{exact}}(\mathcal{O}_{\text{actual}})^{-1}\langle \psi |} = \epsilon
\end{equation}
for all gates $\mathcal{O}$, then for $\epsilon << \frac{1}{\sqrt{\#}}$, \# the number of gates in the circuit, our analysis will be little affected by including error. For larger $\epsilon$, the analysis will fail. This threshold highlights the importance of the first caveat, i.e.\ using $\U_i$ whose powers compile efficiently and keeping the number of gates proportional to $\log k_{\text{max}}$ rather than $k_{\text{max}}$.

\subsection{Amplitude Amplification}\label{sec:amplitude}
In the preceding example, a bias of $\epsilon$ (4\% in the example) towards +1 (or small period) eigenvalues implies $\approx \epsilon^{-2}$ (about 600) runs for a signal to begin to be detected. If our quantum computer is sufficiently capable to run the appropriate variant of amplitude amplification \cite{bh97} then, up to a log factor, the number of runs required reduces to $O(\epsilon^{-1})$. In the context of section \ref{sec:with-structure} this essentially doubles the number of qubits that can be managed.

Just prior to the measurement step in Figure \ref{fig:number-register} we may write to the state as:
\begin{equation}
	|\Psi\rangle = \frac{1}{\sqrt{2^n}} \sum_{k=1}^{2^n} |k\rangle \otimes \left(\alpha_k | \vec{0} \rangle + \beta_k | \psi_1^k \rangle\right),
\end{equation}
where we have decomposed $\U^k(|\vec{0}\rangle) =: \alpha_k |\vec{0}\rangle + \beta_k | \psi_1^k\rangle$. Let
\begin{equation}
	P := \left(\frac{1}{\sqrt{2^n}} \sum_{k=1}^{2^n} | k \rangle \otimes | \vec{0} \rangle \right) \left(\frac{1}{\sqrt{2^n}} \sum_{k=1}^{2^n} \langle k | \otimes \langle \vec{0} |\right)
\end{equation}
be the projection onto $\Theta := \frac{1}{\sqrt{2^n}} \sum_{k=1}^{2n} |k\rangle \otimes |\vec{0}\rangle$, and define $Q = - S_\Psi S_P$, where
\begin{equation}
	S_\Psi := 1 - 2 |\Psi\rangle\langle\Psi|,\ S_P := 1 - 2P
\end{equation}

Similar to the case of \cite{bh97}, where $P$ is projection onto the span of the ``good,'' i.e.\ marked states, iteration of the unitary $Q$ effects a rotation by the angle $2\te$ in the plane spanned by $\Theta$ and $\Psi$, where $\sin \te$ is the overlap of the initial state $\Psi$ with $\Theta$.

In the large $n$ limit, it follows from \eqref{eq:haar-random} that $\sin \te \approx \te \approx \epsilon$, where $\epsilon$ is the bias towards +1 eigenvalues for $\U$, so $O(\epsilon^{-1})$ iteration of $Q$ suffices to align $\psi$ to be detectable by a $|\vec{0}\rangle$-direction measurement in the state register. Since too many iterations of $Q$ will cause an over-shoot, as with Grover's algorithm, a geometric sequence of integral ``guesses'' for $\epsilon^{-1}$ must be tried. This leads to the aforementioned log factor.

Amplitude amplification may also be applied to the less structured problem of section \ref{sec:tensor-decomp}, but does not remove its exponential scaling, similarly for the doubly exponential scaling in the problem of section \ref{sec:haar-random}.

\section{Complexity and Circuit Obfuscation}\label{sec:comp-dif}
The recurrence algorithm allows a NISQ quantum computer to distinguish some exceptional classes of quantum circuits---those resulting in many periodic eigenvalues, as in our tensor example from section \ref{sec:with-structure}. In the next section, we will give context and motivation for (hidden) tensor product circuits and those enjoying further additional spectral structures. In this section, we ask if this application (from section \ref{sec:with-structure}) of the recurrence algorithm surpasses classical resources, i.e.\ does it constitute an example of ``quantum supremacy?'' To put this question in context, we consider the complexity of some related quantum decision problems.

Obfuscation is a difficult and much-researched question. Its simplest instance involves ``trivial'' quantum circuits, that is, quantum circuits which compute $\id: \C^{2n} \ra \C^{2n}$, the identity. An early result \cite{jwb05} is that verification that a classical description of a quantum circuit specifies (nearly) the identity (up to an overall phase) is QMA-complete. The authors build a circuit $Z$ which approximates the identity iff an input quantum circuit $W$ rejects all its potential witnesses with high probability. This result still leaves open the question: is there an efficient classical probabilistic algorithm for constructing quantum circuits $\{C_i\}$ equal (or near) $\id$ which cannot be classically (efficiently, and with high probability) recognized to have that property, given access to the circuit $C_i$? There are many related notions of obfuscation in both classical and quantum CS; in this paper we say such an algorithm $A$ obfuscates the identity map. The existence of such an $A$ is open.

Extending \cite{jwb05}, we find that the language, we call No Unitary Spectral Gap (NUSG), consisting of classical descriptions of quantum circuits having an eigenvalue near 1, is QMA-complete. A precise statement and proof of this result is presented in Appendix \ref{sec:qma-complete}. This shows that even the easiest promise-provisions regarding unitary spectra are difficult. Of course, a property that is difficult to determine (near $\id$ or no spectral gap) is, by no means, necessarily easy to obfuscate. Appendix \ref{sec:qma-complete} merely establishes difficulty.

Obfuscation of classical algorithms is a topic in classical cryptography. \cite{barak01} provided a no-go theorem to a strong black-box form of the classical question, but the same year the lead author published a paper on practical evasions of the black-box barrier \cite{barak01b}.

With the development of quantum computer science, the BQP-version of obfuscation, both of classical and quantum programs, has been studied \cite{af16}. While extending the black-box no-go theorem, they propose a variety of weaker quantum obfuscation schemes and find applications of these to  ``quantum money'' as one example.

Important work \cite{mah18} introduced a technique later exploited in \cite{bm21} in which it is shown that, in the presence of several plausible cryptographic conjectures, it is at least possible to obfuscate, in the quantum context, what they call the ``null function,'' i.e.\ a function which is exponentially unlikely to return anything except 0 on a random input. Their techniques may be applicable to the open problem of obfuscating\footnote{The existence of such an obfuscation technique is at least consistent with the QMA-completeness result of \cite{jwb05}, and the extension we prove in Appendix \ref{sec:qma-complete}.} $\id: \C^{2^n} \ra \C^{2^n}$.

Let us give some heuristic intuition, without theoretical guarantees, that detecting triviality from a circuit description is difficult classically, with two thought experiments. The first thought experiment: Start with any circuit which is manifestly the identity, say one of the form $\U \U^\dagger$, and then cover the circuit with some random (and secrete) pattern of overlapping blocks. Each block, $B_i$, should contain several gates from the original circuit, and each gate be present in several blocks. Now recompile, using distinct primitives, each block $B_i$ to obtain a new block $B_i^\pr$, perhaps with more gates, which effects approximately the same unitary between input and output. Sequentially replacing blocks yields a circuit near the identity, but it would seem difficult to determine this.\footnote{\cite{mwg21}, pointed out to me by the referee, exploits a similar strategy.}

A second thought experiment is best expressed in Nielsen's continuous model of quantum computation \cite{gdn08}, also see \cite{bfls23}. In this model, a computation is an arc in $\operatorname{SU}(2^n)$ beginning at $\id$. The challenge to the classical observer (of the arc's classical description) is to tell if this arc is (nearly) a loop, i.e.\ also ends back at the identity. Following Nielsen, we should choose a left-invariant Riemannian metric $g_{ij}$ on $\operatorname{su}(2^n)$ which penalizes high body interaction, such as the metric written in \eqref{eq:bs-metric} with constant $\gg 1$. The indices $1 \leq i,j \leq 4-1$ run over all non-trivial Pauli-strings of length $n$, where in each slot a Pauli matrix, 1, $X$, $Y$, or $Z$ resides. The trivial all-1 string is not in the Lie algebra $\operatorname{su}(2^n)$. An attractive and now well-studied choice, the Brown-Susskind metric \cite{bs18}, may be written diagonally on the Lie algebra as:
\begin{equation}\label{eq:bs-metric}
	g_{ij} = e^{\text{const.} w(i)} \delta_{i,j},
\end{equation}
where we have written $g_{ij}$ in the Pauli word basis (up to factors of $i$ needed to make the Pauli strings skew-Hermitian), and the weight $w(i)$ of the $i$th Pauli word is, by definition, the number of $X$, $Y$, or $Z$ entries. $w(i)$ measures the ``many-bodiness'' of an interaction. The intuition is that the metric should impose a growing penalty as the number of bodies increases corresponding to the difficulty of engineering such interactions. As computed in \cite{bfls23}, this metric has directions of high negative curvature. The question of ``loop or arc'' becomes a control problem in the presence of negative curvature (and high dimensions). In accordance with the philosophy promoted by Stephen Wolfram \cite{wolfram02}, almost all systems of this type are ``irreducible''; there will usually be no way to tell what they will do, close or not close, without integrating them. Computational short-cuts do not often appear.

Actually, two (or more) versions of classical recognition of trivial quantum circuits can be considered.
\begin{enumerate}
	\item Consider, in a fixed model, all quantum circuits, with an appropriate measure $\mu$. Pick a norm to define a circuit as ``$\epsilon$-close'' to $\id$ and let $\mu_\epsilon$ be the restricted measure. Question: Can a classical agent efficiently distinguish random samples from $\mu_{\frac{\epsilon}{10}}$ from random samples of $\mu \setminus \mu_\epsilon$ (the constant 10 being arbitrary)?

	\item Is there a classical protocol for constructing circuits in $\operatorname{supp}(\mu_{\frac{\epsilon}{10}})$, or even circuits exactly defining the identity, and a second protocol for producing circuits far from $\id$ so that no classical agent can do better than chance at distinguishing the two outputs?
\end{enumerate}

To summarize, the purpose of this section has been to argue, from the literature, but heuristically, that classical discrimination in either of the two contexts above is likely to be difficult, and obfuscation may well be possible.

\section{Hidden Structure in Dynamically Produced Linear Algebra}\label{sec:hidden-structure}

This section digresses from the quantum world to consider tensor decomposition and spectral problems analogous to those treated in sections \ref{sec:rec-alg} and \ref{sec:comp-dif}, but now in a machine learning (ML) context. We are now pursuing a mathematical analogy; we are not proposing here any quantum advantage over classical ML.

Quite broadly, two themes dominate much of science and engineering: minimize a non-linear loss function and linearize (produce the Hessian) around the solution. This mantra becomes a prescription for constructing ``highly trained'' linear operators. Prominent today are the $Q,K,V$ maps generated in the training of ``attention heads'' of LLMs, and also the weight matrices $W$ linking layers of feed forward neural nets (NN). I will predict that for many decades, we will be occupied by the question: With training, do such matrices acquire special structures? Already widely observed (see \cite{gmrm23}) are $\frac{1}{\text{rank}}$ scaling for the largest eigenvalues of convolutional weight matrices. Any structure found, by definition, permits both compression and understanding of the stored information, so as part of our effort to find what is essential in our optimized (i.e.\ trained) linear maps, we will inevitably look for structure.

Let us return to hidden tensor structures, but now in the context of inner-product spaces (a linear space $\R^n$ or $\C^n$ with a non-singular bilinear or sesquilinear pairing $\langle,\rangle$). Hidden tensor structures may be entirely described as a spectral property, as explained shortly, but except in a few cases this spectral property looks intractable (perhaps NP-complete) to verify.

Before giving the mathematical setup, note that the contexts where people have found tensor structures in emergent linear maps between inner product spaces ranges from linguistics \cites{smolensky90,mlds19} to high energy physics \cites{fsz21,fsz21b,fsz21c,szbf23}.

If one picks a functional of some simplicity, it should not be surprising that when optimized, the solution itself should possess a lot of symmetry. For example, it is now known \cite{ckmd22} that almost any sensible function for assessing ``optimality'' of a packing in $\R^8$ or $\R^{24}$ will settle on the highly structured lattices $E_8$ and Leech, respectively. In much the same way, it was found in the papers cited above that for a variety of functionals (including the Ricci scalar) on left invariant metrics on the special unitary group for a small number of qubits ($\operatorname{SU}(4)$ and $\operatorname{SU}(8)$ were studied extensively), local minima were surprisingly likely to exhibit an emergent tensor structure decomposing their eigenspectra $\in i \ast \operatorname{su}(4) = \{4 \times 4 \text{ traceless Hermitians}\}$ (and on $i \ast \operatorname{su}(8) = \{8 \times 8 \text{ traceless Hermitians}\}$), as two (and three) fold tensor products of $2 \times 2$ Hermitian matrices.

At the other end of the intellectual universe, Smolensky and co-authors have been finding iterated, tree-like tensor decompositions ``roles and fillers'' applicable to language modeling \cites{smolensky90,mlds19}. To give some sense on how syntactic roles and fillers are modeled as vectors in tensor product spaces, we reproduce an example from \cite{mlds19} and \cite{myz13}. They consider the ``equation''
\begin{equation}\label{eq:word-ex}
	\text{I see now} - \text{I see} = \text{you know now} - \text{you know}
\end{equation}
If the roles are placed numbers: $|1\rangle$, $|2\rangle$, $|3\rangle$, and the fillers are the words in the place, then with a left to right word ordering \eqref{eq:word-ex} does become a valid equation in a tensor product space:
\begin{multline}
	\Big(|\text{I}\rangle \otimes |1\rangle + |\text{see}\rangle \otimes |2\rangle + |\text{now}\rangle \otimes |3\rangle\Big) - (|\text{I}\rangle \otimes |1\rangle + |\text{see}\rangle \otimes |2\rangle) \\
	= \Big(|\text{you}\rangle \otimes |1\rangle + |\text{know}\rangle \otimes |2\rangle + |\text{now}\rangle \otimes |3\rangle\Big) - (|\text{you}\rangle \otimes |1\rangle + |\text{know}\rangle \otimes |2\rangle)
\end{multline}
However, if the roles are numbered right to left, \eqref{eq:word-ex} becomes false:
\begin{multline}
	\Big(|\text{I}\rangle \otimes |3\rangle + |\text{see}\rangle \otimes |2\rangle + |\text{now}\rangle \otimes |1\rangle\Big) - (|\text{I}\rangle \otimes |2\rangle + |\text{see}\rangle \otimes |1\rangle) \\
	\neq \Big(|\text{you}\rangle \otimes |3\rangle + |\text{know}\rangle \otimes |2\rangle + |\text{now}\rangle \otimes |1\rangle\Big) - (|\text{you}\rangle \otimes |2\rangle + |\text{know}\rangle \otimes |1\rangle)
\end{multline}

With these two disparate examples in mind (high energy and linguistics), let us formulate what is meant by a linear operator between inner product spaces having a (hidden) tensor decomposition and the spectral implications.

\begin{figure}[ht]
	\centering
	\begin{tikzpicture}[scale=1.2]
		\node at (-3.5,3.8) {Left invariant metrics};
		\node at (-3.5,3.3) {on $\operatorname{su}(2^n)$ minimizing};
		\node at (-6,1.3) {$R_{\text{Ricci}} = -\frac{1}{4}$};
		\node at (-4,1.3) {\includegraphics[scale=0.6]{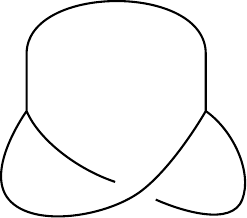}};
		\node at (-2.7,1.3) {$-\frac{1}{2}$};
		\node at (-1.2,1.3) {\includegraphics[scale=0.54]{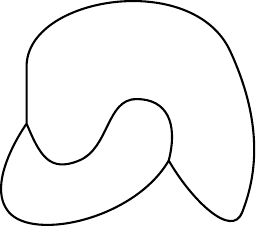}};
	
		\node at (-3.7,-0.5) {$= -\frac{1}{2} c_{ij}^k c_{i^\pr j^\pr}^{k^\pr}g^{ii^\pr}g^{jj^\pr}g_{kk^\pr} - \frac{1}{2}c_{ij}^k c_{ik}^j g^{ii^\pr}$};
	
		\draw (-5.3,-1.7) -- (-5.3,-2) -- (-5.5,-2.2);
		\draw (-5.3,-2) -- (-5.1,-2.2);
		\node at (-3.3,-2) {$= c_{ij}^k$, $[h_i,h_j] = c_{ij}^k h_k$};
		\draw (-5.6,-3) to[out=60,in=120] (-4.6,-3);
		\node at (-4.1,-2.9) {$g_{ij}$,};
		\draw (-3.5,-2.7) to[out=-60,in=-120] (-2.5,-2.7);
		\node at (-2.1,-2.9) {$g^{ij}$};

		\node at (4,3.8) {Tree-like structure};
		\node at (4,3.3) {of language};
	
		\draw (2.5,-1) -- (2.5,0.5) -- (5,0.5) -- (5,-1) -- cycle;
		\draw (5,0.5) -- (5.5,1) -- (5.5,-0.5) -- (5,-1);
		\draw (2.5,0.5) -- (3,1) -- (5.5,1);
		\node at (3.75,-0.25) {LLM};
	
		\draw (2.6,-3) -- (3.4,-0.7) -- (4.1,-0.7) -- (3.3,-3) -- cycle;
		\node[rotate=71] at (3.35,-1.85) {language in $\rightarrow$};
		\draw (4.1,1) -- (4.9,2.9) -- (5.6,2.9) -- (4.8,1) -- cycle;
		\node[rotate=68] at (4.8,1.95) {next token $\rightarrow$};
	\end{tikzpicture}
	\caption{Q: What do emergence of qubits in high energy theory and tree-like structure of language have in common? A: Emergent tensor structures. These examples motivate the problem of detecting such decompositions.}
\end{figure}
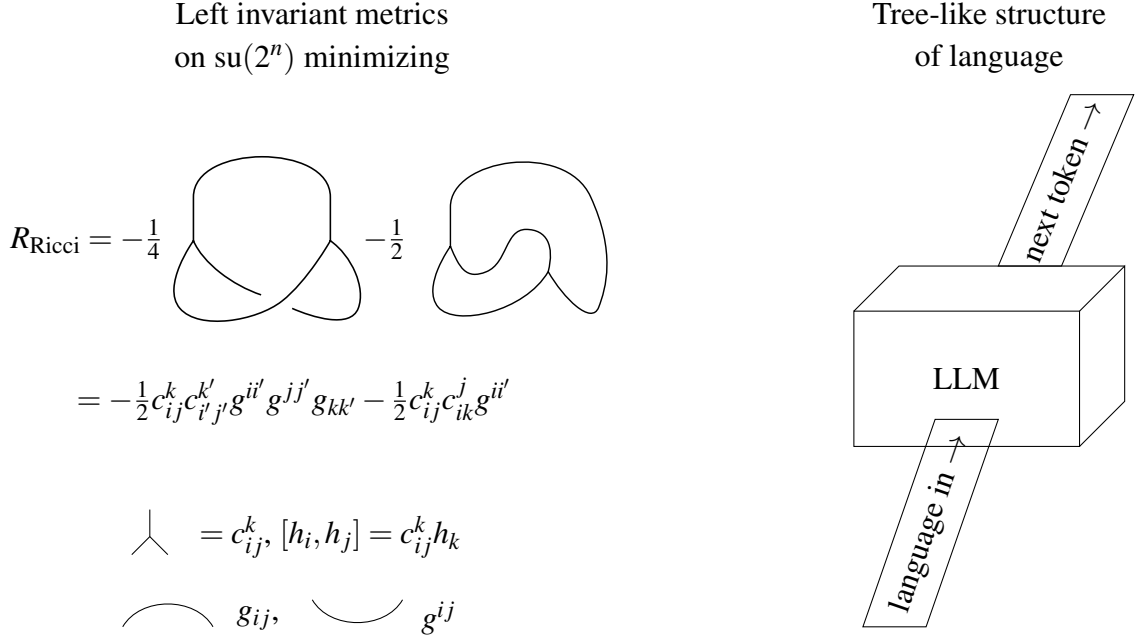

We begin by reviewing the singular value decomposition. There are two settings: a linear map $M: W \ra Y$ where (1) $W$ and $Y$ are finite dimensional vector spaces of $\R$, each with a non-singular bilinear form (inner product) $\langle,\rangle$, or (2) $W$ and $Y$ are finite dimensional vector spaces over $\C$, each with a non-singular Hermitian inner product, also denoted by $\langle,\rangle$. The first case is natural in NNs and language models, the second case is more natural in quantum applications. Since the two are parallel, we treat only the second. To obtain the first case from the second, merely replace all unitary groups with the corresponding orthogonal group. The map $M$ may be written as
\begin{equation}
	M = \U \Sigma V^\dagger,
\end{equation}
where $\U \in \U(Y)$, $V \in \U(W)$, the respective groups of Hermitian isometries, and $\Sigma$ is a diagonal matrix with non-negative real entries (in some chosen orthonormal basis for $W$ and $Y$; different choices of basis merely conjugate the unitaries). These entries are the \emph{singular values}. Note that the inner product $\langle,\rangle$ is required to make sense of the notion that $\U$ and $V$ are unitary. By definition, unitary transformations are linear maps which preserve the Hermitian inner product $\langle,\rangle$.

If $M_1: W_1 \ra Y_1$ and $M_2: W_2 \ra Y_2$ are tensored, it is immediate that
\begin{equation}
	M_1 \otimes M_2 = (\U_1 \otimes \U_2)(\Sigma_1 \otimes \Sigma_2)(V_1^\dagger \otimes V_2^\dagger),
\end{equation}
i.e.\ the SVDs tensor. More generally, if there are $r$ tensor factors,
\begin{equation}\label{eq:r-tensors}
	M_1 \otimes \cdots \otimes M_r = (\U_1 \otimes \cdots \otimes \U_r)(\Sigma_1 \otimes \cdots \otimes \Sigma_r)(V_1^\dagger \otimes \cdots \otimes V_r^\dagger).
\end{equation}

Although $\Sigma_1 \otimes \Sigma_2$ is diagonal, in the tensor basis it is helpful to also consider as a rectangular array of non-negative reals, where the $(i,j)$ entry is $\sigma_{1,i} \sigma_{2,j}$, $\{\sigma_{1,i}\}$ and $\{\sigma_{2,j}\}$ being the singular values of $M_1$ and $M_2$, respectively. More generally, the diagonal entries of $\Sigma_1 \otimes \cdots \otimes \Sigma_r$ may be pictured as lying in an $r$-dimensional rectangular solid with typical entry:
\begin{equation}\label{eq:typical-entry}
	(i_1,\dots,i_r)\text{-entry of } \Sigma_1 \otimes \cdots \otimes \Sigma_r = \sigma_{1,i_1} \cdots \sigma_{r,i_r}
\end{equation}

From \eqref{eq:r-tensors} it is clear that if $M,M^\pr: W \ra Y$ have the same (multi-)set of singular values, then $M = XM^\pr Z$ for appropriate unitary $X \in \U(Y)$ and $Z \in \U(W)$.

It follows that if the singular values of $M$ biject with an array with the form of \eqref{eq:typical-entry}, then $M$ admits a (hidden) tensor decomposition of that same form, and vice versa. Hidden tensor decomposition is a spectral question.

Notice that it is highly unlikely for a random spectrum to assume any block rectangular form, as in \eqref{eq:typical-entry}, beyond the trivial case where all but one tensor factor is of dimension one. That is, the multiplicative expression \eqref{eq:typical-entry} is highly over-determined. This is made even clearer by focusing on $\{\log(\text{singular values})\}$.

\subsubsection*{Observation} $M: W^n \ra Y^n$ can be built as $M = A(M_1 \otimes \cdots \otimes M_r)B^\dagger$, where $M_i: W_i^{n_i} \ra Y_i^{m_i}$, $\prod_{i=1}^r n_i = n$, and $\prod_{i=1}^r m_i = m$, with $A: W^n \ra \otimes_{i=1}^n W^{n_i}$ and $B: Y^n \ra \otimes_{i=1}^r Y^{n_i}$, iff the (possibly repeated) $\log(\text{singular values})$ of $M$ can be used to fill an $n_1 \times \cdots \times n_r$ array so that the additive analog of \eqref{eq:typical-entry} holds, i.e., the $r$-axis of the arrays are labeled by $\log(\text{singular values})$ of $M_i$, $1 \leq i \leq r$, and each array entry is the sum of the values on its $r$ coordinates.

This observation creates a bridge, explored in Appendix \ref{sec:bridge}, between tensor factorization and the Kolmogorov-Arnold representability theorem, the mathematical basis of the well-known assertion \cite{hn87} that given arbitrary width, any continuous function can be represented by a neural net with only one hidden layer.

More immediately, the observation allows us to set up a permutation dependent over-determined system of linear equations with the singular values of $M$ as constants, whose solution is equivalent to factoring $M$ as $A(M_1 \otimes \cdots \otimes M_r)B^\dagger$, $A$ and $B$ unitary. For example, set $r=2$, select variables $x_1,\dots,x_{k_1}$ and $y_1,\dots,y_{k_2}$, where $k_i = \min(n_i,m_i)$, $i \in \{1,2\}$, and set $k = k_1k_2$. The equations for the logarithmic version of \eqref{eq:typical-entry} become:
\begin{equation}
	x_i + y_j = \log(\sigma_{p(i,j)}),\ 1 \leq i \leq k_1 \text{, and } 1 \leq j \leq k_2,
\end{equation}
where $\sigma$ denotes a singular value of $M$ and $p: \{1,\dots,k_1\} \times \{1,\dots,k_2\} \ra \{1,\dots,k\}$ is an unknown bijection. The letter $p$ is used to evoke a choice of permutation. For general $r$, the equations are:
\begin{equation}\label{eq:tensor-problem}
	x_{i_1} + \cdots + x_{i_r} = \log(\sigma_{p(i_1,\dots,i_r)}),\ 1 \leq i_j \leq k_j,\ 1 \leq j \leq r,\ k = \prod_{i=1}^r k_i,
\end{equation}
and $p: \{(1,\dots,n_{i_1}\} \times \cdots \times \{1, \dots, n_{i_r}\} \ra \{1,\dots,n\}$ is an unknown bijection. Given a collection of $k$ $\log(\text{singular values})$, we have the \emph{tensor factor problem} (with format $(k_1,\dots,k_r)$): is there a bijection $p$ as on \eqref{eq:tensor-problem} so that the variables (on the lhs) may be set to non-negative real values to obtain a solution. This problem comes in both \emph{exact} and myriad \emph{approximate} versions where there is some error budget. In the latter case, the admissible error might be some $\epsilon$ per each equation, or a rms budget for all errors, or some fraction of equations that must be solved to some tolerance, etc.

For a fixed $p$, exact solution, if possible, is easy. With $p$ still fixed there is a literature on \emph{learning with errors} (LWE), often over finite fields, with some cases of approximate solution conjectured to be difficult \cite{regev09}. But the larger difficulty is in the unknown ``permutation'' $p$. Such hidden permutations are the basis of many famous NP-complete problems, such as the traveling salesman problem\footnote{On the other hand, the graph isomorphism problem, which is widely believed \emph{not} to be NP-complete, is also masked by a hidden permutation.} and likely makes many versions of the tensor factoring problem NP-complete.

In joint work to appear \cite{fbt24}, various greedy algorithms are studied which work well when the spectrum of singular values decays rapidly. In such cases, one may well guess that the largest singular values should fit into the extreme (most positive) corner of the rectangular solid.

Any special spectral feature which aids the detection of tensor decomposisions is of interest. In the positive example of section \ref{sec:with-structure}, the recurrence algorithm succeeds when applied to tensor products with such an additional feature.

We close this section by locating the concept ``hidden tensor decomposition'' as a special case of ``hidden tensor rank'' (HTR) and ``unitary hidden tensor rank'' (UHTR).

For the problem of decomposing a unitary $\U: \mathcal{H}^{2n} \ra \mathcal{H}^{2n}$ into a (hidden) tensor form:
\begin{equation}\label{eq:unitary-tensor}
	\U = V(\U_1 \otimes \cdots \otimes \U_r)V^\dagger,\ U_i: \mathcal{H}^{2^{n_i}} \ra \mathcal{H}^{2^{n_i}},\ \sum_{i=1}^r n_i = n.
\end{equation}

Similar to our discussion of singular values, the decomposition \eqref{eq:unitary-tensor} exists iff the $2^n$ eigenvalues of $\U$ can fill a $2^{n_1} \times \cdots \times 2^{n_r}$ rectangular solid array so that \eqref{eq:typical-entry} holds. (Since these eigenvalues are not positive reals, any logarithm introduces indeterminacy, so we do not take one.)

Since the spectra of tensor products (both in the context of singular values and eigenvalues) enjoy so much internal structure, the existence of hidden tensor decomposition may be a good target for quantum algorithms beyond the recurrence testing described in section \ref{sec:rec-alg}. We leave this for future consideration.

Before leaving the subject of tensor decomposition of operators, let us point out that a decomposition
\begin{equation}
	M = T(M_1 \otimes M_2)V^\dagger
\end{equation}
says that $M$ has ``rank 1'' but \emph{not} in the usual sense of matrix rank, but in a sense we call ``hom-tensor-rank'' (HTR), defined below. In a unitary context, we may require that the isometries $T$ and $V$ agree, $T=V$. This reduced freedom leads to the notion of unitary hom-tensor-rank (UHTR). We explain.

$M: W \ra Y$ is a functional on $W$ with values in $Y$, hence $M \in W^\ast \otimes Y$. The familiar meaning of ``$M$ has matrix rank 1'' ($\equiv$ ``$M$ has row rank 1'' $\equiv$ ``$M$ has column rank 1'') is that there exists a $w^\ast \in W^\ast$ and a $y \in Y$ so that $M = w^\ast \otimes y$. Recall the general element of $W^\ast \otimes Y$ has the form $\sum_{i=1}^r w_i^\ast \otimes y_i$, and the minimal $r$ required to write $M$ in this way is matrix rank$(M)$, as commonly defined. That is, the familiar meaning of rank is the minimal number of tensor summands necessary to describe $M$.

But now, consider the general operator between inner product spaces with fixed tensor factorings: $M: W \ra Y$, where $(W, \langle,\rangle_W) \cong (W_1 \otimes W_2, \langle,\rangle_{W_1} \cd \langle,\rangle_{W_2})$, i.e.\ $\langle w_1 \otimes w_2, \tld{w}_1 \otimes \tld{w}_2\rangle_W = \langle w_1, \tld{w}_1\rangle_{W_1} \cd \langle w_2,\tld{w}_2\rangle_{W_2}$, under the isomorphism, and similarly for $Y$.
\begin{equation}\label{eq:tensor-factors}
	\begin{split}
	M \in (W_1 \otimes W_2)^\ast \otimes (Y_1 \otimes Y_2) & \cong W_1^\ast \otimes W_2^\ast \otimes Y_1 \otimes Y_2 \\
	& \cong W_1^\ast \otimes Y_1 \otimes W_2^\ast \otimes Y_2 \\
	& \cong \operatorname{Hom}(W_1,Y_1) \otimes \operatorname{Hom}(W_2,Y_2)
	\end{split}
\end{equation}
$M$ can always be written as $M = \oplus_{i=1}^r M_{1,i} \otimes M_{2,i}$, $M_{1,i}: W_1 \ra Y_1$ and $M_{2,i}: W_2 \ra Y_2$. We define HTR to be the minimum such $r$ as we vary over possible isometries $W_1 \otimes W_2 \ra W$ and $Y_1 \otimes Y_2 \ra Y$ which may be parametrized by elements of the unitary groups $\U(W)$ and $\U(Y)$, respectively. UHTR is defined in the same way except that since $T=V$, the minimization is over a single isometry $H_1 \otimes H_2 \leftrightarrow H$. In particular, if for appropriate isometries $M$ decomposes as $M = M_1 \otimes M_2$, $M_1 \in \operatorname{Hom}(W_1,Y_1)$ and $M_2 \in \operatorname{Hom}(W_2,Y_2)$, then $M$ has HTR$(M) = 1$. HTR comes from the ordinary tensor rank w.r.t.\ the tensor indicated on the last line of \eqref{eq:tensor-factors}.\footnote{UHTR should not be confused with Krauss rank (see Theorems 8.1, 8.2, and 8.3 of \cite{nc00}). Even in the case that $M$ is a completely positive super-operator, they are quite distinct. For example, conjugation by $\U$ on $\operatorname{Hom}(H,H) \cong H^\ast \otimes H$ has Krauss rank one and may be written as $\sum_{h^\ast,h^\pr} h^\ast \otimes \U^\dagger h^\ast \otimes h^\pr \otimes \U h^\pr$, the sum over bases for $H^\ast$ and $H$. But UHTR (conjugation by $\U$) is the rank over the middle $\otimes$, which, generically, is large.}

Matrix rank of operator is the bread and butter of applied math; one often reduces dimension by truncating to low rank approximation. Matrix rank is as easy to compute as SVDs. What makes HTR$(M)$ difficult to compute is that it may be hidden, perhaps even in a cryptographic sense, by the possible coordinate change inherent in the action of the two unitary groups. For example, HTR$(M) = 1$ iff it can be written as:
\begin{equation}
	M = \U(M_1 \otimes M_2) V^\dagger.
\end{equation}
where $\U \in \U(W)$ and $Y \in \U(Y)$.

\bibliography{references}

\begin{bibdiv}
\begin{biblist}

\bib{af16}{article}{
      author={Alagic, Gorjan},
      author={Fefferman, Bill},
       title={On quantum obfuscation},
        date={2016},
      eprint={arXiv:1602.01771},
}

\bib{a13}{inproceedings}{
      author={Ambainis, Andris},
       title={On physical problems that are slightly more difficult than
  {QMA}},
        date={2014},
   booktitle={{IEEE 29th Conference on Computational Complexity}},
       pages={32\ndash 43},
}

\bib{arn56}{article}{
      author={Arnold, Vladimir},
       title={On functions of three variables},
        date={1956},
     journal={{Proceedings of the USSR Academy of Sciences}},
      volume={114},
       pages={679\ndash 681},
}

\bib{barak01b}{inproceedings}{
      author={Barak, Boaz},
       title={How to go beyond the black-box simulation barrier},
        date={2001},
   booktitle={Proceedings 42nd {IEEE Symposium on Foundations of Computer
  Science}},
       pages={106\ndash 115},
}

\bib{fbt24}{article}{
      author={Barkeshli, Maissam},
      author={Freedman, Michael},
      author={He, Tianyu},
        date={2024},
        note={Work in progress},
}

\bib{bfls23}{article}{
      author={Brown, Adam~R.},
      author={Freedman, Michael~H.},
      author={Lin, Henry~W.},
      author={Susskind, Leonard},
       title={Universality in long-distance geometry and quantum complexity},
        date={2023},
     journal={Nature},
      volume={622},
      number={7981},
       pages={58\ndash 62},
}

\bib{bfs11}{article}{
      author={Brown, Brielin},
      author={Flammia, Steven~T.},
      author={Schuch, Norbert},
       title={Computational difficulty of computing the density of states},
        date={2011},
     journal={Physical Review Letters},
      volume={107},
      number={4},
       pages={040501},
}

\bib{barak01}{inproceedings}{
      author={Barak, Boaz},
      author={Goldreich, Oded},
      author={Impagliazzo, Rusell},
      author={Rudich, Steven},
      author={Sahai, Amit},
      author={Vadhan, Salil},
      author={Yang, Ke},
       title={On the (im)possibility of obfuscating programs},
        date={2001},
   booktitle={{Annual International Cryptology Conference}},
       pages={1\ndash 18},
}

\bib{bh97}{inproceedings}{
      author={Brassard, Gilles},
      author={Hoyer, Peter},
       title={An exact quantum polynomial-time algorithm for {S}imon's
  problem},
        date={1997},
   booktitle={{Proceedings of the Fifth Israeli Symposium on Theory of
  Computing and Systems}},
       pages={12\ndash 23},
}

\bib{bm21}{article}{
      author={Bartusek, James},
      author={Malavolta, Giulio},
       title={Indistinguishability obfuscation of null quantum circuits and
  applications},
        date={2021},
      eprint={arXiv:2106.06094},
}

\bib{brattka07}{inproceedings}{
      author={Brattka, Vasco},
       title={From {H}ilbert’s 13th problem to the theory of neural networks:
  {C}onstructive aspects of {K}olmogorov’s {S}uperposition {T}heorem},
        date={2007},
   booktitle={{Kolmogorov's Heritage in Mathematics}},
      editor={\'{E}ric Charpentier},
      editor={Lesne, Annick},
      editor={Nikolski, Nickola\"{i}},
       pages={253\ndash 280},
}

\bib{bs18}{article}{
      author={Brown, Adam~R.},
      author={Susskind, Leonard},
       title={Second law of quantum complexity},
        date={2018},
     journal={Physical Review D},
      volume={97},
      number={8},
       pages={086015},
}

\bib{ckmd22}{article}{
      author={Cohn, Henry},
      author={Kumar, Abhinav},
      author={Miller, Stephen~D.},
      author={Radchenko, Danylo},
      author={Viazovska, Maryna},
       title={Universal optimality of the {$E_8$} and {L}eech lattices and
  interpolation formulas},
        date={2022},
     journal={Annals of Mathematics},
      volume={196},
      number={3},
       pages={983\ndash 1082},
}

\bib{cs13}{book}{
      author={Conway, John~H.},
      author={Sloane, Neil~J.A.},
       title={Sphere packings, lattices and groups},
   publisher={Springer Science \& Business Media},
        date={2013},
}

\bib{d21}{article}{
      author={Dzhenzher, Sviatoslav},
       title={A simpler proof of {S}ternfeld's theorem},
        date={2021},
     journal={Journal of Topology and Analysis},
}

\bib{fsz21}{article}{
      author={Freedman, Michael},
      author={Shokrian-Zini, Modjtaba},
       title={The universe from a single particle},
        date={2021},
     journal={Journal of High Energy Physics},
      number={140},
}

\bib{fsz21b}{article}{
      author={Freedman, Michael},
      author={Shokrian-Zini, Modjtaba},
       title={The universe from a single particle {II}},
        date={2021},
     journal={Journal of High Energy Physics},
      number={102},
}

\bib{fsz21c}{article}{
      author={Freedman, Michael},
      author={Shokrian-Zini, Modjtaba},
       title={The universe from a single particle {III}},
        date={2021},
      eprint={arXiv:2112.08613},
}

\bib{fan17}{article}{
      author={Fan, Ruihua},
      author={Zhang, Pengfei},
      author={Shen, Huitao},
      author={Zhai, Hui},
       title={Out-of-time-order correlation for many-body localization},
        date={2017},
     journal={Science Bulletin},
      volume={62},
      number={10},
       pages={707\ndash 711},
}

\bib{gdn08}{article}{
      author={Gu, Mile},
      author={Doherty, Andrew},
      author={Nielsen, Michael},
       title={Quantum control via geometry: An explicit example},
        date={2008},
     journal={Physical Review A},
      volume={78},
      number={3},
       pages={032327},
}

\bib{gmrm23}{article}{
      author={Guth, Florentin},
      author={M\'{e}nard, Brice},
      author={Rochette, Gaspar},
      author={Mallat, St\'{e}phane},
       title={A rainbow in deep network black boxes},
        date={2023},
      eprint={arXiv:2305.18512},
}

\bib{hn87}{inproceedings}{
      author={Hecht-Nielsen, Robert},
       title={Kolmogorov’s mapping neural network existence theorem},
        date={1987},
   booktitle={{Proceedings of the International Conference on Neural
  Networks}},
      volume={3},
       pages={11\ndash 14},
}

\bib{jwb05}{article}{
      author={Janzing, Dominik},
      author={Wocjan, Pawel},
      author={Beth, Thomas},
       title={{``Non-identity-check''} is {QMA}-complete},
        date={2005},
     journal={International Journal of Quantum Information},
      volume={3},
      number={3},
       pages={463\ndash 473},
}

\bib{kol57}{article}{
      author={Kolmogorov, Andrey},
       title={On the representation of continuous functions of several
  variables by superpositions of continuous functions of a smaller number of
  variables},
        date={1957},
     journal={Proceedings of the USSR Academy of Sciences},
      volume={108},
       pages={179\ndash 182},
}

\bib{ksv02}{book}{
      author={Kitaev, Alexei},
      author={Shen, Alexander},
      author={Vyalyi, Michael~N.},
       title={Classical and quantum computation},
      series={Graduate Studies in Mathematics},
   publisher={American Mathematical Society},
        date={2002},
      volume={47},
}

\bib{lsyz21}{article}{
      author={Lu, Jianfeng},
      author={Shen, Zuowei},
      author={Yang, Haizhao},
      author={Zhang, Shijun},
       title={Deep network approximation for smooth functions},
        date={2021},
     journal={SIAM Journal on Mathematical Analysis},
      volume={53},
      number={5},
       pages={5465\ndash 5506},
}

\bib{mah18}{inproceedings}{
      author={Mahadev, Urmila},
       title={Classical verification of quantum computations},
        date={2018},
   booktitle={{2018 IEEE 59th Annual Symposium on Foundations of Computer
  Science}},
       pages={259\ndash 267},
}

\bib{mlds19}{inproceedings}{
      author={McCoy, R.~Thomas},
      author={Linzen, Tal},
      author={Dunbar, Ewan},
      author={Smolensky, Paul},
       title={{RNN}s implicitly implement tensor product representations},
        date={2019},
   booktitle={{The Seventh International Conference on Learning
  Representations}},
}

\bib{mwg21}{misc}{
      author={Maslov, Dmitri},
      author={Wocjan, Pawel},
      author={Gambetta, Jay~Michael},
       title={Obfuscation of quantum circuits},
        date={2021},
        note={U.S. Patent 2023/0196155A1},
}

\bib{myz13}{inproceedings}{
      author={Mikolov, Tom{\'a}{\v{s}}},
      author={Yih, Wen-tau},
      author={Zweig, Geoffrey},
       title={Linguistic regularities in continuous space word
  representations},
        date={2013},
   booktitle={{Proceedings of the 2013 Conference of the North American Chapter
  of the Association for Computational Linguistics}},
       pages={746\ndash 751},
}

\bib{nc00}{book}{
      author={Nielsen, Michael~A.},
      author={Chuang, Isaac~L.},
       title={Quantum computation and quantum information},
   publisher={Cambridge University Press},
        date={2000},
}

\bib{regev09}{article}{
      author={Regev, Oded},
       title={On lattices, learning with errors, random linear codes, and
  cryptography},
        date={2009},
     journal={Journal of the ACM},
      volume={56},
      number={6},
       pages={1\ndash 40},
}

\bib{shalizi06}{inproceedings}{
      author={Shalizi, C.R.},
       title={Methods and techniques of complex systems science: An overview},
        date={2006},
   booktitle={{Complex Systems Science in Biomedicine}},
      series={Topics in Biomedical Engineering International Book Series},
   publisher={Springer},
       pages={33\ndash 114},
}

\bib{smolensky90}{article}{
      author={Smolensky, Paul},
       title={Tensor product variable binding and the representation of
  symbolic structures in connectionist systems},
        date={1990},
     journal={Artificial Intelligence},
      volume={46},
      number={1-2},
       pages={159\ndash 216},
}

\bib{ss06}{inproceedings}{
      author={Skopenkov, Arkadiy},
      author={Shnurnikov, Igor},
       title={Basic planar sets},
        date={2006},
   booktitle={{18th Summer Conference of the International Mathematical
  Tournament of Towns}},
}

\bib{szbf23}{article}{
      author={Shokrian-Zini, Modjtaba},
      author={Brown, Adam~R.},
      author={Freedman, Michael},
       title={The smallest interacting universe},
        date={2023},
     journal={Journal of High Energy Physics},
      number={82},
}

\bib{takens06}{inproceedings}{
      author={Takens, Floris},
       title={Detecting strange attractors in turbulence},
        date={1981},
   booktitle={{Dynamical Systems and Turbulence, Warwick 1980: Proceedings of a
  Symposium Held at the University of Warwick 1979/1980}},
      series={Lecture Notes in Mathematics},
      volume={898},
       pages={366\ndash 381},
}

\bib{error-function}{misc}{
      author={Wikipedia},
       title={Error function},
        date={2024},
        note={In \emph{Wikipedia}},
}

\bib{wolfram02}{book}{
      author={Wolfram, Stephen},
       title={A new kind of science},
   publisher={Wolfram Media},
        date={2002},
}

\end{biblist}
\end{bibdiv}

\appendix
\section{No Unitary Spectral Gap (NUSG) is QMA-Complete}\label{sec:qma-complete}
\begin{definition}[QMA]
	Fix $\epsilon = \epsilon(\abs{x}) \leq \frac{1}{3}$, but no smaller than exponentially small in $\abs{x}$, $x \in \{0,1\}^\ast$. A language $L$ is in QMA if for all $x \in \{0,1\}^\ast$, $x$ efficiently dictates a ``verifying quantum circuit'' $\U_x$. $\U_x$ has an input register that can hold a (possibly mixed) witness state $\rho$ and an ancilla register. $\U_x$ must satisfy:
	\begin{enumerate}
		\item If $x \in L$, $\exists \rho$ s.t.\ $\operatorname{tr}(\U_x(\rho \otimes | 0 \cdots 0 \rangle \langle 0 \cdots 0 | \U_x^\dagger P_1)) \geq 1 - \epsilon$, and
		\item If $x \notin L$, $\forall \rho$, $\operatorname{tr}(\U_x(\rho \otimes | 0 \cdots 0 \rangle \langle 0 \cdots 0 | \U_x^\dagger P_1)) \leq \epsilon$,
	\end{enumerate}
	where $P_1$ is the projector equivalent to ``first qubit of input register is in state $|1\rangle$.'' The allowed choices of $\epsilon$ all yield the same computational class.

	We say that in case 1, with witness $\rho$, $\U_x$ accepts with probability $\geq 1-\epsilon$, and in case 2, for all witnesses $\rho$, $\U_x$ rejects with probability $\geq 1 - \epsilon$. (It is sufficient to consider $\rho = |\psi\rangle\langle\psi|$, pure.)
\end{definition}

\begin{definition}[No Unitary Spectral Gap NUSG$(\delta_0)$]
	This is the language of classical descriptions for quantum circuits $\U_x$ with the property $\operatorname{spec}(\U_x) \cap \{e^{i\delta} \mid -\delta_0 < \delta < \delta_0\} \neq \varnothing$ for fixed $0 < \delta_0 < 0.01$ ($\delta_0$ should be small but the upper bound 0.01 is arbitrary).
\end{definition}

\begin{definition}[NUSG Problem]
	Let $x \in \{0,1\}^\ast$ be an efficient classical description of quantum circuit $\U_x$. The no unitary spectral gap problem is: Given the premise that either $x \in \operatorname{NUSG}(\delta_0)$ or $x \notin \operatorname{NUSG}(10 \delta)$ (the constant 10 is arbitrary), determine which alternative holds.
\end{definition}

\begin{theorem}\label{thm:nusg-qma}
	$\operatorname{NUSG}(\delta_0)$ is QMA-complete.
\end{theorem}

The proof is an adaptation of the proof that ``identity check'' is QMA-complete \cite{jwb05}; we have tried to keep our notation parallel to theirs. Like their proof, our direct construction only yields a small probabilistic separation between ``accept'' and ``reject,'' but \cite{ksv02} details how this separation may be amplified.

\begin{proof}
	Given appropriate eigenwitnesses, the quantum phase estimation algorithm approximates its eigenvalues, showing NUSG is in QMA.

	To check completeness, we use the following circuit; compare to Figure 2 of \cite{jwb05}.

	\begin{figure}[ht]
		\centering
		\begin{tikzpicture}
			\draw (-6,-4) -- (6,-4);
			\draw (-6,-2.5) -- (6,-2.5);
			\draw (-6,-1.75) -- (6,-1.75);
			\draw (-6,-0.25) -- (6,-0.25);
			\draw (-6,1.25) -- (6,1.25);
			\draw (-6,2) -- (6,2);
			\draw (-7,3.5) -- (6,3.5);
			\node at (-7.4,3.5) {$| 0 \rangle$};
		
			\draw[fill=white] (-4,2.5) -- (-2,2.5) -- (-2,-4.5) -- (-4,-4.5) -- cycle;
			\node at (-3,-1) {$\U_x$};
			\draw[fill=white] (4,2.5) -- (2,2.5) -- (2,-4.5) -- (4,-4.5) -- cycle;
			\node at (3,-1) {$\U_x^\dagger$};
			\draw (-4.7,3) -- (-4.7,-4);
			\draw[fill=black] (-4.7,-1.75) circle (0.4ex);
			\draw[fill=black] (-4.7,-2.5) circle (0.4ex);
			\draw[fill=black] (-4.7,-4) circle (0.4ex);
			
			\node at (-5.3,-3.2) {$\vdots$};
			\node at (0,-3.2) {$\vdots$};
			\node at (5,-3.2) {$\vdots$};
			\node at (-5.3,0.5) {$\vdots$};
			\node at (0,0.5) {$\vdots$};
			\node at (5,0.5) {$\vdots$};
			\node at (-6.5,0.4) {$|\psi\rangle$};
			\node at (-6.5,-4) {$| 0 \rangle$};
			\node at (-6.5,-3.2) {$\vdots$};
			\node at (-6.5,-1.75) {$| 0 \rangle$};
			\node at (-6.5,-2.5) {$| 0 \rangle$};
			
			\draw[fill=white] (-5.5,4.25) -- (-3.5,4.25) -- (-3.5,2.75) -- (-5.5,2.75) -- cycle;
			\node at (-4.5,3.5) {$V$};
			\draw (0,2.8) -- (0,2);
			\draw[fill=black] (0,2) circle (0.4ex);
			\node at (-4.7,-5) {$V$};
			\node at (-3,-5) {$\U$};
			\node at (0,-5) {$Y$};
			\node at (3,-5) {$\U^\dagger$};
		
			\draw[fill=white] (-1.5,4.25) -- (1.5,4.25) -- (1.5,2.75) -- (-1.5,2.75) -- cycle;
			\node at (0,3.5) {\footnotesize{$C\begin{pmatrix} e^{i\phi} & 0 \\ 0 & e^{-i\phi} \end{pmatrix}$}};
			\draw[fill=white] (2.7,4) -- (3.3,4) -- (3.3,3) -- (2.7,3) -- cycle;
			\node at (3,3.5) {$H$};
			\draw[fill=white] (-6.6,4) -- (-6,4) -- (-6,3) -- (-6.6,3) -- cycle;
			\node at (-6.3,3.5) {$H$};
		
			\node[rotate=90] at (-8.3,3.5) {ancilla};
			\node[rotate=90] at (-8.3,0.4) {input};
			\node[rotate=90] at (-7.9,0.4) {witness};
			\node[rotate=90] at (-8.1,-3) {ancillas};
		\end{tikzpicture}
		\caption{$Z := H \U^\dagger Y \U V H$.}\label{fig:completeness}
	\end{figure}

	By definition, $V$ induces the operator
	$\begin{pmatrix}
		e^{-i\phi} & 0 \\[-0.75em] 0 & e^{i\phi}
	\end{pmatrix}$
	on the top ancilla when the lower ancillas are precisely in state $|0 \cdots 0\rangle$ and an overall phase shift
	$\begin{pmatrix}
		e^{2i\phi} & 0 \\[-0.75em] 0 & e^{2i\phi}
	\end{pmatrix}$
	for the other (basic) ancilla settings. The unitary $Y$ is controlled by the first qubit of the input register and applies the operator
	$\begin{pmatrix}
		e^{i\phi} & 0 \\[-0.75em] 0 & e^{-i\phi}
	\end{pmatrix}$
	exactly when this first qubit is in state $|1\rangle$, the accept state for $\U_x$. The angle $\phi$ should be small. It will be chosen later so as not to be too small with respect to the $\epsilon>0$ in the definition of QMA. We will need $\phi \gg \sqrt{\epsilon}$.

	The first input qubit of $\U_x$ is used to determine acceptance, indicated by the state $|1\rangle$, of witness $\psi$ with ancillas set in the state $|0 \cdots 0 \rangle$. Let $P_1$ be the projection onto the $|1\rangle$-state of this qubit and $P_0 := 1 - P_1$ be the complementary projection to $|0\rangle$.

	Define
	\begin{equation}
		| \Psi \rangle = |0 \rangle \otimes | \psi \rangle \otimes | 0 \cdots 0 \rangle,
	\end{equation}
	and assume $| \psi \rangle$ is accepted by $\U$ (case 1) with probability $\geq 1-\epsilon$; we then compute the action of the circuit on $\Psi$:
	\begin{align*}
		Z | \Psi \rangle & = H \U^\dagger Y \U V (H(| 0 \rangle) \otimes | \psi \rangle \otimes |0 \cdots 0 \rangle) \\
		& = H \U^\dagger Y \U \left(\frac{1}{\sqrt{2}}(e^{-i\phi}|0\rangle + e^{i\phi}|1 \rangle) \otimes | \psi \rangle \otimes | 0 \cdots 0 \rangle\right) \\
		& = H \U^\dagger Y (1-P_0) \U \left(\frac{1}{\sqrt{2}}(e^{-i\phi}|0\rangle + e^{i\phi}|1 \rangle) \otimes | \psi \rangle \otimes | 0 \cdots 0 \rangle\right) \\
		& \quad + H \U^\dagger Y P_0 \U \left(\frac{1}{\sqrt{2}}(e^{-i\phi}|0\rangle + e^{i\phi}|1 \rangle) \otimes | \psi \rangle \otimes | 0 \cdots 0 \rangle\right) \\
		& = H \U^\dagger (1-P_0) \U \left(\frac{1}{\sqrt{2}}(|0\rangle + |1 \rangle) \otimes | \psi \rangle \otimes | 0 \cdots 0 \rangle\right) \\
		& \quad + H \U^\dagger Y P_0 \U \left(\frac{1}{\sqrt{2}}(e^{-i\phi}|0\rangle + e^{i\phi}|1 \rangle) \otimes | \psi \rangle \otimes | 0 \cdots 0 \rangle\right) \\
		& = H \left(\frac{1}{\sqrt{2}}(|0\rangle + |1 \rangle) \otimes | \psi \rangle \otimes | 0 \cdots 0 \rangle\right) \\
		& \quad - H \U^\dagger P_0 \U \left(\frac{1}{\sqrt{2}}(|0\rangle + |1 \rangle) \otimes | \psi \rangle \otimes | 0 \cdots 0 \rangle\right) \\
		& \quad + H \U^\dagger P_0 \U \left(\frac{1}{\sqrt{2}}(e^{-i\phi}|0\rangle + e^{i\phi}|1 \rangle) \otimes | \psi \rangle \otimes | 0 \cdots 0 \rangle\right)
		=: | \Psi \rangle - | \phi_1 \rangle + | \phi_2 \rangle
	\end{align*}
	The acceptance probability of $|\psi \rangle$ implies norm$| \phi_i \rangle \leq \sqrt{\epsilon}$, $i = 1,2$, so
	\begin{equation}\label{eq:accept-prob}
		\lVert Z| \Psi \rangle - |\Psi \rangle \rVert \leq 2 \sqrt{\epsilon}.
	\end{equation}

	Now consider case 2, where $\U_x = \U$ accepts the general $|\psi\rangle \otimes |0 \cdots 0 \rangle$ with probability $\leq \epsilon$. We need to consider the general input $|\Phi\rangle$ to $Z$. It is convenient to decompose $| \Phi \rangle = | \Phi_1 \rangle \oplus | \Phi_2 \rangle$ where $\Phi_1$ is in the span of states with ancilla qubits set to $|0 \cdots 0 \rangle$ and $\Phi_2$ is in the span with the other ancilla basis states. We call these subspaces in the first and second summands, respectively.
	\begin{equation}
		Z|\Phi\rangle = Z |\Phi_1\rangle + Z|\Phi_2\rangle.
	\end{equation}

	Let us evaluate $Z | \Phi_1 \rangle$ first.
	\begin{equation}
		H \U^\dagger Y \U V H |\Phi_1\rangle = H\U^\dagger Y P_1 \U V H |\Phi_1\rangle + H\U^\dagger YP_0 \U V H |\Phi_1\rangle
	\end{equation}
	Noting $Y P_0 = P_0$, we obtain:
	\begin{equation}
		\begin{split}
		Z(\Phi_1) & = H\U^\dagger Y P_1 \U V H | \Phi_1 \rangle + H \U^\dagger P_0 \U V H | \Phi_1 \rangle \\
		& = H\U^\dagger Y P_1 \U V H | \Phi_1 \rangle + H V H | \Phi_1 \rangle - H \U^\dagger P_1 \U V H | \Phi_1 \rangle
		\end{split}
	\end{equation}
	$\lVert P_1 \U V H | \Phi_1 \rangle \rVert \leq \sqrt{\epsilon} \lVert | \Phi_1 \rangle \rVert$, since the probability of acceptance is $\leq \epsilon$. Thus:
	\begin{equation}
		\lVert Z | \Phi_1 \rangle - V H | \Phi_1 \rangle \rVert \leq 2 \sqrt{\epsilon} \lVert |\Phi_1 \rangle \rVert
	\end{equation}
	and note that $(VH - H)$ has norm $\sin \phi \approx \phi$ on the first summand.

	Using line \eqref{eq:accept-prob} on this first summand,
	\begin{equation}
		\lVert Z | \Phi_1 \rangle - | \Phi_1 \rangle \rVert \geq \sin \phi - 2 \sqrt{\epsilon} \lVert | \Phi_1 \rangle \rVert.
	\end{equation}
	
	Now consider $Z | \Phi_2 \rangle$. On this second summand $\lVert VH - H \rVert = \lVert V - \id \rVert$, from the definition of $V$ is equal to $\sin 2\phi \approx 2\phi$. Regardless of how $\U$ outputs onto its first qubit the overall $e^{i2\phi}$ phase (from $V$) overwhelms any phases from $Y$ (note output $|1\rangle$ does not necessarily mean ``acceptance'' on the second summand since there the ancilla are not set to $|0 \cdots 0 \rangle$). The result is that on the second summand, we also have
	\begin{equation}
		\lVert Z | \Phi_2 \rangle - | \Phi_2 \rangle \rVert \geq \sin \phi - 2 \sqrt{\epsilon} \lVert | \Phi_2 \rangle \rVert.
	\end{equation}
	The last two lines imply the desired spectral gap: in case 2, for general normalized $| \Phi \rangle$ we have:
	\begin{equation}
		\lVert Z | \Phi \rangle - | \Phi \rangle \rVert \geq \sin \phi - 2 \sqrt{\epsilon}
	\end{equation}
	Since $\phi$ can be chosen small but still satisfying $\sqrt{\epsilon} << \phi$, the spectral gap in case 2 is confirmed. To measure this theoretical gap, the output $Z | \Phi \rangle$ from copies of the circuit in Figure \ref{fig:completeness} should be compared with a fresh copy of $|\Phi\rangle$ using the SWAP test, which efficiently estimates $\abs{\langle \Phi | Z | \Phi \rangle}$. Thus, any black box capable of certifying membership in NUSG can also certify $x \in L$.
\end{proof}

\begin{note}
	Theorem \ref{thm:nusg-qma} exemplifies one of many QMA-complete spectral problems. Verifying the size of a spectral gap is generally DQMA-complete \cite{a13}, a presumably harder status than QMA-complete, and estimating spectral degeneracies even harder, being complete for \#QMA \cite{bfs11}. The latter problem related both to our approach to recurrence detection and detection of the tensor structure.
\end{note}

\section{A Curious Bridge Between Tensor Factoring and the Kolmogorov-Arnold Representation Theorem}\label{sec:bridge}

Hilbert's 13th problem was on what he thought was the impossibility of expressing the local motion of a root $r$ of a high degree polynomial as an iterated concatenation of 1-variable functions of the coefficients which are allowed to be combined into multivariable functions only by ``plus'' (``times'' is implicitly okay to use as well since $x \cd y = \exp(\log x + \log y)$). The quadratic formula
\[
	r = \frac{-b \pm \sqrt{b^2 - 4ac}}{2a}
\]
has this structure, as do the closed form solutions in degrees 3 and 4.

Hilbert's intuition was based on ``dimension count'': It feels like there are too many 2-variable functions to write them all out in terms of 1-variable functions. As Hilbert knew, this dimension counting argument can be made rigorous in the world of analytic functions; however, Hilbert stated his problem in the continuous setting. There, as Kolmogorov \cite{kol57} and Arnold \cite{arn56} showed, his intuition was mistaken. The result, stated in one of many possible forms, is given below. To understand how low regularity can confound dimensional intuition, it is useful to recall the ``space filling curve'' $\beta: [0,1] \twoheadrightarrow [0,1]^2$. $\beta$ is defined by ``coordinate splitting'' the decimal expansion of each $x \in [0,1]$, e.g.\ $\beta(0.00110101\dots) = (0.0100\dots,0.0111\dots)$. This map is far from analytic but is H\"{o}lder-$\frac{1}{2}$ continuous. Similar, but more ingenious devices, underlie the proof of Theorem \ref{thm:kol-arn}.

\begin{theorem}[Kolmogorov-Arnold (KA) representation theorem]\label{thm:kol-arn}
	Given $n \geq 2$, $\exists n(2n+1)$ continuous functions $\phi_{p,q}: I \ra I$, $1 \leq p \leq n$, $1 \leq q \leq 2n+1$, such that for continuous $f: I^n \ra I$, there exists continuous $\alpha_q: I \ra I$, $1 \leq q \leq 2n+1$ such that $f(x_1,\dots,x_n) = \sum_{q=1}^{2n+1} \alpha_q (\sum_{p=1}^n \phi_{p,q}(x_p))$. The $\phi_{p,q}$ can be chosen to be Lipschitz.\footnote{The Lipschitz property is subtle as Lipschitz curves are never space filling, see \cite{brattka07}.}
\end{theorem}

\begin{figure}[ht]
	\centering
	\begin{tikzpicture}[scale=1.1]
		\draw[fill=black] (-4,1) circle (0.2ex);
		\node at (-4.4,1) {$x_1$};
		\draw[fill=black] (-4,-1) circle (0.2ex);
		\node at (-4.4,-1) {$x_2$};

		\draw (-4,1) -- (0,2);
		\draw (-4,1) -- (0,1);
		\draw (-4,1) -- (0,0);
		\draw (-4,1) -- (0,-1);
		\draw (-4,1) -- (0,-2);

		\node at (-2,1.8) {$\phi_{p,q}$};
		
		\draw[dashed] (-4,-1) -- (0,0);
		\draw[dashed] (-4,-1) -- (0,-1);
		\draw[dashed] (-4,-1) -- (0,-2);
		\draw[dashed] (-4,-1) -- (0,1);
		\draw[dashed] (-4,-1) -- (0,2);

		\draw[fill=black] (0,2) circle (0.2ex);
		\draw[fill=black] (0,1) circle (0.2ex);
		\draw[fill=black] (0,0) circle (0.2ex);
		\draw[fill=black] (0,-1) circle (0.2ex);
		\draw[fill=black] (0,-2) circle (0.2ex);

		\node at (2,1.3) {$\alpha_q$};

		\draw (0,2) -- (4,0);
		\draw (0,1) -- (4,0);
		\draw (0,0) -- (4,0);
		\draw (0,-1) -- (4,0);
		\draw (0,-2) -- (4,0);

		\draw[fill=black] (4,0) circle (0.2ex);
		\node at (4.9,0) {$f(x_1,x_2)$};
	\end{tikzpicture}
	\caption{Hidden layer}\label{fig:hidden-layer}
\end{figure}

As Figure \ref{fig:hidden-layer} makes clear, Kolmogorov and Arnold were inventing what we now call neural networks and demonstrating their expressivity. But it took 30 years for the interpretation to emerge \cite{hn87}.

The functions $\{\phi_{p,q}\}$, describe a universal Lipschitz embedding of $I^n$ into $I^{2n+1}$ with the remarkable property that any function on $I^n$ is merely the sum of $2n+1$ ``activation'' functions on the $2n+1$ coordinate projections. An embedding with this property is called \emph{basic}. Even more remarkable, this universal embedding is far from unique; basic embeddings, from the proof, form an infinite-dimensional space.

One caveat: in early formulations the $f$-dependent activation functions $\{\alpha_q\}$ are extremely complicated. However, later work \cite{lsyz21} shows that by increasing the width of a network with a \emph{single} hidden layer from $2n+1$ towards infinity, simple activation functions such as ReLu suffice to approximate any continuous function from $I^n$. They also show how the approximation becomes more efficient as the regularity of $f$ increases. Of course, approximation is still more efficient (and more learnable) if the network is deep rather than shallow.

The ``obstruction'' to an embedding being basic is the existence within its image of certain sets of points called \emph{Sternfeld arrays}, which are either closed or of increasing size. The definition is in general a bit tricky, see \cite{d21} for a modern discussion. But, the discrete analog of a Sternfeld array for a hyper-cubulated hyper-rectangular solid has been discussed in \cite{ss06} and is remarkably well-aligned with the tensor factorization problem discussed in section \ref{sec:hidden-structure}, wherein we sought to find the location of $\{\log(\text{singular values})\}$ in a hyper-rectangular array.

By \emph{definition}, a discrete Sternfeld array (DSA) is a minimal collection of sites in a hyper-cubical lattice with the property that Real values can be assigned to these cubes that \emph{cannot} be reconstructed by summing any collection of real values placed along the coordinate axes. Dually, we may say a DSA is a minimal set of locations inside the hyper-rectangle which supports a signed measure $\mu$ whose push-forward to each coordinate axis vanishes identically.\footnote{The push-forward of measure under proejction to \emph{all} 1-dimensional subspaces is the Radon transform, $\mathcal{R}$. A DSA $X$ is characterized by the property $\mathcal{R}_{e_i}(f) = 0$ for all functions supported on $X$, and $\{e_i,1 \leq i \leq n\}$ the basis for $\R^n$.}

Figure \ref{fig:dsa} is a simple 2D example of a four point DSA with the atomic measure values indicated.

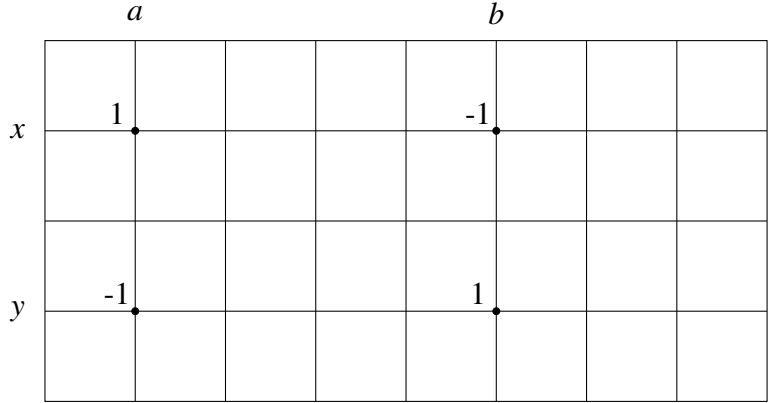
\begin{figure}[ht]
	\centering
	\begin{tikzpicture}[scale=1.2]
		\draw (-4,2) -- (4,2) -- (4,-2) -- (-4,-2) -- cycle;
		\draw (-2,2) -- (-2,-2);
		\draw (-3,2) -- (-3,-2);
		\draw (-1,2) -- (-1,-2);
		\draw (0,2) -- (0,-2);
		\draw (1,2) -- (1,-2);
		\draw (2,2) -- (2,-2);
		\draw (3,2) -- (3,-2);
		\draw (-4,1) -- (4,1);
		\draw (-4,0) -- (4,0);
		\draw (-4,-1) -- (4,-1);
	
		\draw[fill=black] (-3,1) circle (0.2ex);
		\draw[fill=black] (-3,-1) circle (0.2ex);
		\draw[fill=black] (1,1) circle (0.2ex);
		\draw[fill=black] (1,-1) circle (0.2ex);
	
		\node at (-3.2,1.2) {1};
		\node at (-3.2,-0.8) {-1};
		\node at (0.8,1.2) {-1};
		\node at (0.8,-0.8) {1};
		\node at (-3,2.3) {$a$};
		\node at (1,2.3) {$b$};
		\node at (-4.3,1) {$x$};
		\node at (-4.3,-1) {$y$};

		\node at (4.3,0) {$\hphantom{e}$};
	\end{tikzpicture}
	\caption{Cardinality 4 DSA with a signed measure drawn on the dual lattice.}\label{fig:dsa}
\end{figure}

Clearly, the measure pushes forward to 0 on both the $x$ and $y$ axes. Dually, regarding the numerical labels as functions values, they cannot be realized by summing coordinate values. We check
\[
	a+x = 1,\ a+y=-1,\ b+x = -1, \text{ and } b+y = 1 \implies x-y = 2 \text{ and } x-y = -2
\]
An exercise in \cite{ss06} is to prove that 2D DSA are precisely the $90^\circ$ turning points of a rook moving in a closed path on this rectangular chess board. (Proof sketch: every horizontal (vertical) segment of the path gives, by subtraction, a constraint on the horizontal (vertical) labels. The closing of the path places a final linear constraint on both the horizontal and vertical labels, over-determining both. If no rook-path closes, the constraints do not over-determine the edge variables, a closed path over-determining these variables.)

A DSA, $X$, acts as both an obstruction to an embedding being basic (in the sense of the KA theorem) and as a constraint on writing a $k$-large (multi-)set of $\log(\text{eigen/singular values})$ as a set-sum of many smaller (multi-)sets of cardinalities $k_1,\dots,k_r$, $k = \prod_{i=1}^r k_i$. Thus, there must be a bridge connecting hidden tensor decompositions (HTD), which are characterized by set-sum structure of their eigen/singular values, and KA. Proposition \ref{prop:r-dsa} formulates this connection.

According to \cite{ss06}, the problem of characterizing DSA in dimensions $\geq 3$ is open. Proposition \ref{prop:2d-dsa} treats the more tractible case of 2D DSAs, which by the preceding exercise, are the WRCs of Definition \ref{def:wrc}. We prove by induction that any collection of more than $p+q-1$ sites in a 2D $p \times q$ rectangle contains a DSA. Mapping this information back into the problem of section \ref{sec:hidden-structure}, decomposing $M$ into $\U(M_1 \otimes M_2)V^\dagger$, $M_1$ ($M_2$) with $p$ ($q$) singular values, we see that rather than being able to solve for consistent locations of all the $pq$ singular values of $M$, only at most $p+q-1$ locations assigned will admit a general solution (i.e.\ a coordinate labeling), and then only if the locations admit no closed rook circuit.

\begin{definition}\label{def:wrc}
	Let $R$ be the $p \times q$ integral grid of points $\{(i,j) \mid 1 \leq i \leq p, 1 \leq j \leq q\}$. Let $S \subset R$ be a subset. Let $\lbar{R}$ be the rectangle which is the convex hull of $R$. A rook circuit (RC) of $S$ is a closed path in $\lbar{R}$ whose turning points are contained in $S$ and each turn is $\pm 90^\circ$. We say $S$ is \emph{without a rook circuit} (WRC) if $S$ has no RCs. More generally, within a higher dimensional discrete rectangle, a subset $S$ that does not contain any DSA is called a WDSA.
\end{definition}

\begin{proposition}\label{prop:2d-dsa}
	If $S$ is WRC, then $\abs{S} \leq p+q-1$, $\abs{S}$ the cardinality of $S$.
\end{proposition}

\begin{note}
	$S_0 := \{(1,j) \mid 1 \leq j \leq q\} \cup \{(i,1) \mid 1 \leq i \leq p\}$ is an example of a WRC $S$ with $\abs{S} = p+q-1$.
\end{note}

\begin{proof}
	Induct first on $\abs{R} = pq$ and within that use a second induction on $k :=$ the number of points $\abs{S \setminus S_0}$, i.e.\ the ``discrepancy'' from the standard example. The base of these inductions is obvious.

	We may WLOG normalize the situation by assuming $(1,1) \in S$. To check this is legitimate, observe that our problem actually has toroidal $\Z_p \times \Z_q$ symmetry if one focuses only on path turning points and ignores the distinction between +$90^\circ$ and -$90^\circ$ turns. Any rook path on the $\Z_p \times \Z_q$-torus, i.e.\ $R$ with periodic boundary conditions, descends to a rook path in $R$, with the property ``closed'' conserved. The map in the other direction, $R \ra \Z_p \times \Z_q$, is inclusion. Under toroidal symmetry, all points of $R$ become equivalent, so if $S \neq \varnothing$ there is no loss of generality assuming $(1,1) \in S$. Assume, for a contradiction, that $\abs{S} \geq p+q$.

	\begin{definition}
		If $s = (i,j) \in S \setminus S_0$, we say $S$ has ``one hot foot'' if exactly one of $(i,1),(1,j) \in S \cap S_0$. In this case, the pair in $S \cap S_0$ is called ``hot,'' and the other ``cold.'' Since $(1,1) \in S$ if both $(i,1)$ and $(1,j)$ belong to $S$, then $S$ has a RC: $\{(1,1),(1,j),(i,j),(i,1),(1,1)\}$.
	\end{definition}

	\begin{lemma}\label{lm:feet}
		Suppose $S$ is WRC and $s \in S \setminus S_0$ has one hot foot $s^\pr$ and cold foot $s^{\pr\pr}$. Setting $S^\pr = (S \setminus s) \cup \{s^{\pr\pr}\}$, then $S^\pr$ is WRC.
	\end{lemma}

	\begin{proof}
		If $S^\pr$ has a RC call it $\alpha$. Clearly, $S^{\pr\pr} := S^\pr \cup S$ has a RC $\beta = \{(1,1),s^\pr,s,s^{\pr\pr},(1,1))$. Note the union $\alpha \cup \beta$ (canceling overlap) would be a (possibly disconnected) RC of $S$, a contradiction.
	\end{proof}

	\begin{figure}[ht]
		\centering
		\begin{tikzpicture}[scale=1.2]
			\draw (-1.5,-1.5) -- (-1.5,1.5) -- (1.5,1.5) -- (1.5,-1.5) -- cycle;
			\draw (-0.5,-1.5) -- (-0.5,1.5);
			\draw (0.5,-1.5) -- (0.5,1.5);
			\draw (-1.5,0.5) -- (1.5,0.5);
			\draw (-1.5,-0.5) -- (1.5,-0.5);
			\node at (0,-2) {$S^\pr$};
			\node at (-1.75,-0.5) {\footnotesize{$s^{\pr\pr}$}};
			\node at (1.75,1.5) {\footnotesize{$s^\pr$}};
		
			\draw[very thick] (-1.5,-1.5) -- (-0.5,-1.5) -- (-0.5,-0.5) -- (-1.5,-0.5) -- cycle;
			\node at (-1,-1) {$\alpha$};
			\draw[fill=black] (-0.5,-1.5) circle (0.4ex);
			\draw[fill=black] (-0.5,-0.5) circle (0.4ex);
			\draw[fill=black] (1.5,1.5) circle (0.4ex);
			\draw[fill=black] (0.5,1.5) circle (0.4ex);
			\draw[fill=black] (-1.5,1.5) circle (0.4ex);
			\draw[fill=black] (-1.5,0.5) circle (0.4ex);
			\draw[fill=black] (-1.5,-0.5) circle (0.4ex);
			\draw[fill=white] (-1.5,-1.5) circle (0.4ex);
		
			\draw (2.5,-1.5) -- (2.5,1.5) -- (5.5,1.5) -- (5.5,-1.5) -- cycle;
			\draw (3.5,-1.5) -- (3.5,1.5);
			\draw (4.5,-1.5) -- (4.5,1.5);
			\draw (2.5,0.5) -- (5.5,0.5);
			\draw (2.5,-0.5) -- (5.5,-0.5);
			\node at (4,-2) {$S^{\pr\pr}$};
			\node at (2.25,-0.5) {\footnotesize{$s^{\pr\pr}$}};
			\node at (5.75,1.5) {\footnotesize{$s^\pr$}};
			\node at (5.75,-0.5) {\footnotesize{$s$}};
		
			\draw[very thick] (2.5,-1.5) -- (3.5,-1.5) -- (3.5,-0.5) -- (5.5,-0.5) -- (5.5,1.5) -- (2.5,1.5) -- (2.5,-1.5);
			\node at (3,-1) {$\alpha$};
			\node at (4,0.2) {$\beta$};
			\draw[fill=black] (3.5,-1.5) circle (0.4ex);
			\draw[fill=black] (3.5,-0.5) circle (0.4ex);
			\draw[fill=black] (5.5,-0.5) circle (0.4ex);
			\draw[fill=black] (5.5,1.5) circle (0.4ex);
			\draw[fill=black] (4.5,1.5) circle (0.4ex);
			\draw[fill=black] (2.5,1.5) circle (0.4ex);
			\draw[fill=black] (2.5,0.5) circle (0.4ex);
			\draw[fill=black] (2.5,-0.5) circle (0.4ex);
			\draw[fill=white] (2.5,-1.5) circle (0.4ex);
		
			\draw (-5.5,-1.5) -- (-5.5,1.5) -- (-2.5,1.5) -- (-2.5,-1.5) -- cycle;
			\draw (-4.5,-1.5) -- (-4.5,1.5);
			\draw (-3.5,-1.5) -- (-3.5,1.5);
			\draw (-5.5,0.5) -- (-2.5,0.5);
			\draw (-5.5,-0.5) -- (-2.5,-0.5);
			\node at (-4,-2) {$S$};
			\node at (-5.75,-0.5) {\footnotesize{$s^{\pr\pr}$}};
			\node at (-2.25,1.5) {\footnotesize{$s^\pr$}};
			\node at (-2.25,-0.5) {\footnotesize{$s$}};
			
			\draw[very thick] (-5.5,-1.5) -- (-4.5,-1.5) -- (-4.5,-0.5) -- (-2.5,-0.5) -- (-2.5,1.5) -- (-5.5,1.5) -- (-5.5,-1.5);
			\draw[fill=black] (-4.5,-1.5) circle (0.4ex);
			\draw[fill=black] (-4.5,-0.5) circle (0.4ex);
			\draw[fill=black] (-2.5,-0.5) circle (0.4ex);
			\draw[fill=black] (-2.5,1.5) circle (0.4ex);
			\draw[fill=black] (-3.5,1.5) circle (0.4ex);
			\draw[fill=black] (-5.5,1.5) circle (0.4ex);
			\draw[fill=black] (-5.5,0.5) circle (0.4ex);
			\draw[fill=white] (-5.5,-1.5) circle (0.4ex);
			\node at (-4.3,0.2) {$\alpha \cup \beta$};
		\end{tikzpicture}
		\caption{Black dots $= S$, white added to $S$ contradicts $S$ WRC.}
	\end{figure}
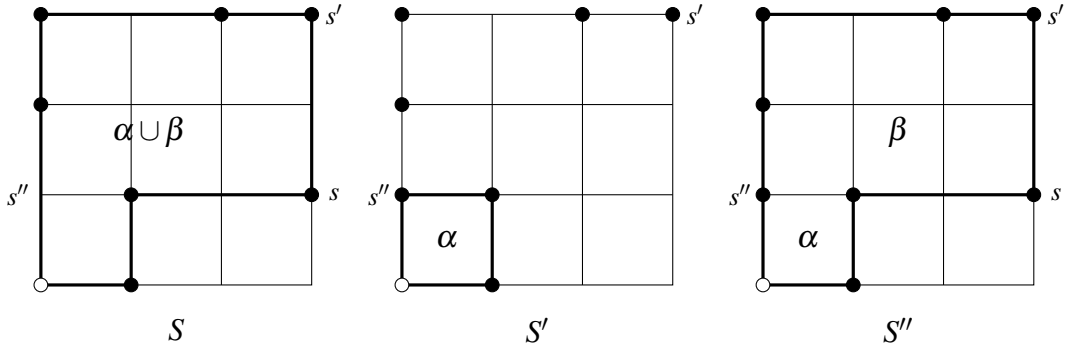

	We conclude for any $S \subset R$, $S$ WRC, and $k(S)$ minimal, that no $s \in S \setminus S_0$ has a hot foot. Consequentially, if $a = \abs{\{1,j) \mid 1 \leq j \leq q\} \setminus S}$ and $b = \abs{\{(i,1) \mid 1 \leq i \leq p\} \setminus S}$, then
	\begin{equation}\label{eq:k-induction}
		k > a + b
	\end{equation}

	Now consider the ``dilute'' sub-$a \times b$-rectangle $AB$ of $R$ made of sites whose vertical (horizontal) projections lie in $\{(1,j) \mid 1 \leq j \leq q\} \setminus S$\ \ ($\{(i,1) \mid 1 \leq i \leq p\} \setminus S$). Since $\abs{AB} < \abs{R} = pq$, the outer induction tells us, in light of Lemma \ref{lm:feet}, that $S \cap (AB) \subset AB$ must have a RC. This contradiction completes the proof.
\end{proof}

According to \cite{ss06}, for $\dim > 2$ (i.e.\ discrete hyper-rectangles $R$), less is known about the geometry or cardinality of a DSA $\subset \Z_{k_1} \times \cdots \times \Z_{k_r}$.

This bridge HTD-KA suggests some weaker structure, a ``partial tensor decomposition,'' in which the operator $M$ is not, itself, found to have a tensor structure, but rather it is embedded in a higher dimensional operator with a tensor structure. Such embeddings are the tensor analog of the universal Kolmogorov-Arnold embedding, $I^n \hookrightarrow I^{2n+1}$.

We have:
\begin{proposition}\label{prop:r-dsa}
	Let $S \subset R = \Z_{k_1} \times \cdots \times \Z_{k_r}$ be an $r$-dimensional WDSA in the (discrete) hyper rectangle $R$, written with periodic boundary conditions. Let $M: W \ra Y$ be a linear map of rank $s = \abs{S}$, $\abs{S}$ the cardinality of $S$, between finite dimensional inner product spaces. Then $M$ may be ``isometrically embedded'' in the higher rank linear operator $\mathcal{O}$ as in the diagram below. All vector spaces here are inner product (or, if one likes, Hermitian inner product) spaces, and the superscripts indicate dimension. ``Isometrically embedded'' is in quotes as it is actually $W \slash \operatorname{ker} M$ and $Y \slash \operatorname{coker} M$ which are isometrically embedded by $\U$ and $V$, respectively. The precise definition is the diagram below:
\end{proposition}

\begin{figure}[ht]
	\centering
	\begin{tikzpicture}
		\node at (0,0) {$A^{k_1} \otimes \cdots \otimes A^{k_r}$};
		\draw[->] (1.5,0) -- (4.5,0);
		\node at (6,0) {$B^{k_1} \otimes \cdots \otimes B^{k_r}$};
		\node at (3,0.2) {\footnotesize{$\mathcal{O}:=\mathcal{O}_1 \otimes \cdots \otimes \mathcal{O}_r$}};
		\node at (0,1.5) {$W \slash \operatorname{ker} M$};
		\draw[->] (1.5,1.5) -- (4.5,1.5);
		\node at (6,1.5) {$Y \slash \operatorname{coker} M$};
		\node at (3,1.7) {\footnotesize{$M$-induced}};
		\draw[->] (0,1.2) -- (0,0.3);
		\node at (0.7,0.75) {\footnotesize{isom. }$\U$};
		\draw[->] (6,1.2) -- (6,0.3);
		\node at (6.7,0.75) {\footnotesize{isom.\ }$V$};
	\end{tikzpicture}
\end{figure}

\begin{proof}
	The DSA Property means that the $s$ $\log(\text{singular values})$ of $M$ may be (arbitrarily) bijected onto the sites of $S$ and the linear equation of \eqref{eq:tensor-problem} with variables at sites of $\Z_{k_1} \coprod \cdots \coprod \Z_{k_r}$, solved. Then $\exp$ of these solutions solve \eqref{eq:typical-entry}. With these solutions in hand, build $\mathcal{O}_i: A^{k_i} \ra B^{k_i}$ to have singular values precisely those occurring as the solutions to \eqref{eq:typical-entry} restricted to the $i$th factor $\Z_{k_i}$. The proposition now follows from the multiplicative behavior of singular values under tensor product.
\end{proof}

Each discrete Sternfeld array enforces some restriction on the eigen/singular values of the large operator, potentially obstructing its tensor factorization. Similarly, in the continuum, Sternfeld arrays obstruct the ability to reconstruct a function from its coordinate projections. In the continuum, functions that factor through a basic embedding are unobstructed in this sense. Proposition \ref{prop:r-dsa} establishes the analogous class of partial tensor decompositions, those unobstructed by DSA. Of course, if $\{\text{singular values }M\}$ has some structure, the equations \eqref{eq:tensor-problem} and \eqref{eq:typical-entry} may still be solvable where $s \gg \abs{\text{WDSA}}$, for any discrete Sternfeld array WDSA.

\end{document}